\documentclass{article}
\usepackage[utf8]{inputenc}

\title{PhD Thesis "Flatness-based Constrained Control and Model-Free Control Applications to Quadrotors and Cloud Computing"\\ \vspace{5mm}
Chapter 2: Constraints on Nonlinear Finite Dimensional Flat Systems}

\author{Author: Maria Bekcheva }
\date{July 2019} 

\usepackage{amsmath}
\usepackage{mathtools}
\usepackage{amssymb}
\usepackage{graphicx, caption, subcaption}
\usepackage{mathrsfs}
\usepackage{hyperref}

\usepackage{appendix}
\usepackage{chngcntr}
\usepackage{etoolbox}
\usepackage{lipsum}
\usepackage{hyperref}

\newtheorem{theorem}{Theorem}
\newtheorem{assumption}{Assumption}

\newtheorem{lemma}{Lemma}
\newtheorem{proposition}{Proposition}
\newtheorem{proof}{Proof}
\newtheorem{question}{Question}
\newtheorem{problem}{Problem}
\newtheorem{definition}{Definition}
\newtheorem{example}{Example}
\newtheorem{remark}{Remark}

\newcommand{\Vect}[1]{{ \mbox{\boldmath ${#1}$} }}
\newcommand{\ie}{{\emph{i.e.}} }
\newcommand{\eg}{{\emph{e.g.}} }

\newcommand{\Reals}{{\mathbb R}}

\newcommand{\Rationals}{{\mathbb Q}}
\newcommand{\Naturals}{{\mathbb N}}

\DeclareRobustCommand \norm[1]{\left\lVert#1\right\rVert}

\makeatletter
\def\munderbar#1{\underline{\sbox\tw@{$#1$}\dp\tw@\z@\box\tw@}}
\makeatother
\newcommand{\rn}{{\mathbb R}}
\newcommand{\nn}{{\mathbb N}}
\global\long\def\dv#1#2{{ \mbox{\boldmath ${#1}$} }^{\langle#2\rangle}}
\global\long\def\dda#1#2{#1^{\langle#2\rangle}}

	\renewcommand{\th}{\textsuperscript{th}} 

\begin{document}
	
	\maketitle
	{\centering
		PhD directors:  Hugues Mounier and Luca Greco \\
 \itshape Universit\'{e} Paris-Saclay \\}
\vspace{15mm}
This is the second chapter in my thesis "Flatness-based Constrained Control and
Model-Free Control Applications to Quadrotors and Cloud Computing". Comments and suggestions are most welcome \footnote{ \url{maria.bekcheva@l2s.centralesupelec.fr} or \url{maria.bekcheva@gmail.com.}}.
\newpage
	\tableofcontents

\newpage

 \textit{Abstract: This chapter presents an approach to embed the input/state/output constraints in a unified manner into the trajectory design for differentially flat systems. To that purpose, we specialize the flat outputs (or the reference trajectories) as Bézier curves. Using the flatness property, the system's inputs/states can be expressed as a combination of Bézier curved flat outputs and their derivatives. Consequently, we explicitly obtain the expressions of the control points of the inputs/states Bézier curves as a combination of the control points of the flat outputs. By applying desired constraints to the latter control points, we find the feasible regions for the output Bézier control points \ie a set of feasible reference trajectories.
}
\section{Chapter overview}
 
 \subsection{Motivation}
The control of nonlinear systems subject to \textit{state and input constraints} is one of the major challenges in control theory.
Traditionally, in the control theory literature, the reference trajectory to be tracked is specified in advance.  Moreover for some applications, for instance, the quadrotor trajectory tracking, selecting the \textit{right trajectory} in order to avoid obstacles while not damaging the actuators is of crucial importance. 

In the last few decades,  Model Predictive Control (MPC) \cite{Carlos1989, Mayne2000} has achieved a big success in dealing with constrained control systems.  Model predictive control is a form of control in which the current control law is obtained by solving, at each sampling instant, a finite horizon open-loop optimal control problem, using the current state of the system as the initial state; the optimization yields an optimal control sequence and the first control in this sequence is applied to the system. It has been widely applied in petro-chemical and related industries where satisfaction of constraints is particularly important because efficiency demands operating points on or close to the boundary of the set of admissible states and controls.

The optimal control or MPC maximize or minimize a defined performance criterion chosen by the user.  The optimal control techniques, even in the case without constraints are usually discontinuous, which makes them less robust and more dependent of the initial conditions. In practice, this means that the delay formulation renders the numerical computation of the optimal solutions difficult.\\

A large part of the literature working on constrained control problems is focused on optimal trajectory generation \cite{Faulwasser2011, VanLoock2015}. These studies are trying to find feasible trajectories that optimize the performance following a specified criterion.  Defining the right criterion to optimize may be a difficult problem in practice. Usually, in such cases, the feasible and the optimal trajectory are not too much different. For example, in the case of autonomous vehicles \cite{LaValle2006}, due to the dynamics, limited curvature, and under-actuation, a vehicle often has few options for how it changes lines on highways or how it travels over the space immediately in front of it. Regarding the complexity of the problem, searching for a feasible trajectory is \textit{easier}, especially in the case where we need \textit{real-time re-planning} \cite{Hagenmeyer2008, Hagenmeyer2010}. Considering that the evolution of transistor technologies is reaching its limits, low-complexity controllers that can take the constraints into account are of considerable interest. The same remark is valid when the system has sensors with limited performance. 

\subsection{Research objective and contribution}

In this chapter, we propose a novel trajectory-based framework to deal with system constraints. We are answering the following question:

\begin{question}
	How to design a set of the reference trajectories (or the feed-forwarding trajectories) of a nonlinear system such that the input, state and/or output constraints are fulfilled?
\end{question} 

\begin{figure}
	\centering \hspace*{-10ex}
	\includegraphics[width=6.5in]{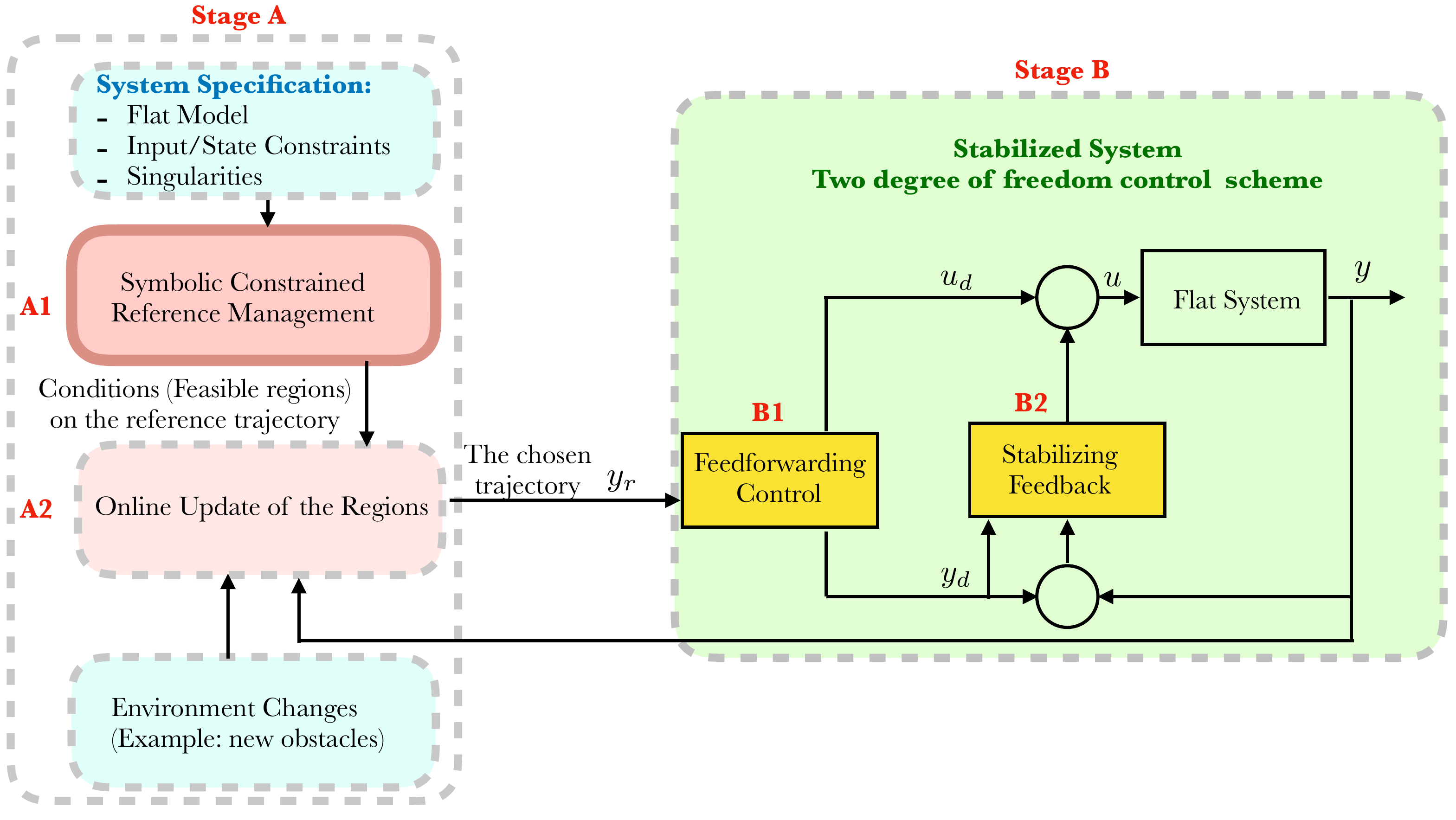}
	\caption{Two degrees of freedom control scheme overview}
	\label{fig_Scheme}
\end{figure}

For that purpose, we divide the control problem in two stages (see Figure \ref{fig_Scheme}). Our objective will be to elaborate a \textit{constrained reference trajectory management} (Stage A) which is meant to be applied to already pre-stabilized systems (Stage B).

Unlike other receding horizon approaches which attempt to solve stabilization, tracking, and constraint fulfilment at the same time, we assume that in \textit{Stage B}, a primal controller has already been designed to stabilize the system which provide nice tracking properties in the absence of constraints. In stage B, we employ the two-degree of freedom design consisting of a constrained trajectory design (constrained feedfowarding) and a  feedback control. 

In \textit{Stage A}, the constraints are embedded in the flat output trajectory design. Thus, \textit{our constrained trajectory generator} defines a feasible open-loop reference trajectory satisfying the states and/or control constraints that \textit{a primal feedback controller} will track and stabilize around. 

To construct Stage A we first take advantage of the \textit{differential flatness property} which serves as a base to construct our method. The differential flatness property yields exact expressions for the state and input trajectories of the system through trajectories of a flat output and its derivatives without integrating any differential equation. The latter property allows us to map the state/input constraints into the flat output trajectory space.

 Then, in our symbolic approach (stage A1), we assign a Bézier curve to each flat output where the parameter to be chosen are the so-called \textit{control points} (yielding a finite number of variables on a finite time horizon) given in a symbolic form. This kind of representation naturally offers several algebraic operations like the sum, the difference and multiplication, and affords us to preserve the explicit functions structure without employing discrete numerical methods. The advantage to deal with the constraints \textit{symbolically}, rather than numerically, lies in that the symbolic solution explicitly depends on the control points of the reference trajectory. This allows to study how the input or state trajectories are influenced by the reference trajectory.

We find symbolic conditions on the trajectory control points such that the states/inputs constraints are fulfilled.

 We translate the state/input constraints into constraints on the reference trajectory control points and we wish to reduce the solution of the systems of equations/inequations into a simpler one. Ideally, we want to find the exact set of solutions \ie the constrained subspace.

We explain how this symbolic constrained subspace representation can be used for \textit{constrained feedforwarding trajectory selection}. The stage A2 can be done in two different ways. 
\begin{itemize}
\item	When a system should track a trajectory in a \textit{static known environment}, then the exact set of feasible trajectories is found and the trajectory is fixed by our choice. If the system's environment changes, we only need to re-evaluate the exact symbolic solution with new numerical values.
 \item When a system should track a trajectory in an \textit{unknown environment with moving objects}, then, whenever necessary, the reference design modifies the reference supplied to a primal control system so as to enforce the fulfilment of the constraints. This second problem is not addressed in the thesis.
\end{itemize}

Our approach is \textit{not based on any kind of optimization} nor does it need computations for a given numerical value at each sampling step. We determine a set of feasible trajectories through the system constrained environment that enable a controller to make quick real-time decisions. For systems with singularities, we can isolate the singularities of the system by considering them as additional constraints.

\subsection{Existing Methods}

\begin{itemize}
		
	\item  Considering actuator constraints based on the derivatives of the flat output (for instance, the jerk \cite{Gasparetto2008,Yu2014}, snap \cite{Mellinger2011}) can be too conservative for some systems. The fact that a feasible reference trajectory is designed following the system model structure allows to choose a quite aggressive reference trajectory.
	\item In contrast to \cite{Lowis2014}, we characterize the whose set of viable reference trajectories which take the constraints into account.
	
	\item  In \cite{Suryawan2012}, the problem of constrained trajectory planning of differentially flat systems is cast into a simple quadratic programming problem ensuing computational advantages by using the flatness property and the B-splines curve's properties. They simplify the computation complexity by taking advantage of the B-spline minimal (resp. maximal) control point. The simplicity comes at the price of having only minimal (resp. maximal) constant constraints that eliminate the possible feasible trajectories and renders this approach conservative.

	\item In \cite{Graichen2008}, an inversion-based design is presented, in which the transition task between two stationary set-points is solved as a two-point boundary value problem. In this approach, the trajectory is defined as polynomial where only the initial and final states can be fixed.
	
	\item  The thesis of Bak  \cite{Bak2000} compared existing methods to constrained controller design (anti-windup, predictive control, nonlinear methods), and introduced a nonlinear gain scheduling approach to handle actuator constraints.

\end{itemize}

\subsection{Outline}
This chapter is organized as follows:
\begin{itemize}
	\item In section \ref{secPreliminaries}, we recall the notions of differential flatness for finite dimensional systems.
	\item In section \ref{sec2ProblemStatement}, we present our problem statement for the constraints fulfilment through the reference trajectory.
	\item In section \ref{SecTraj}, we detail the flat output parameterization given by the Bézier curve, and its properties.
	\item In section \ref{SecProcedure}, we give the whole procedure in establishing reference trajectories for constrained open-loop control. We illustrate the procedure through two applications in section \ref{sec5App}. 
	\item In section \ref{sec4Feasibility}, we present the two methods that we have used to compute the constrained set of feasible trajectories.

\end{itemize}

\section{Differential flatness  overview}\label{secPreliminaries}

The concept of \textit{differential flatness} was introduced in \cite{Fliess1995,Fliess1999} for non-linear finite dimensional systems. By the means of differential flatness, a non-linear system can be seen as a controllable linear system through a dynamical feedback.

A model shall be described by a differential system as:

\begin{align}
\label{eqFlatnessFinite}
\Vect{\dot x} = f(\Vect{x}, \Vect{u})        
\end{align}
where  $ \Vect{x} \in \Reals^n$ denote the \textit{state variables} and
$\Vect{u} \in \Reals^m$ the \textit{input vector}.
Such a system is said to be  \textit{flat} if there exists a set of \textit{flat outputs (or linearizing outputs)} (equal in number to the number of inputs) given by
\begin{equation}
\Vect{y} = h(\Vect{x}, \Vect{u},\Vect{\dot u},..., \Vect{u}^{(r)})
\end{equation}
with $r \in \mathbb{N}$ such that the components of  $\Vect{y} \in \Reals^m$  and all their derivatives are
functionally independent and such that we can parametrize every solution $(\Vect{x}, \Vect{u})$ of \eqref{eqFlatnessFinite} in some dense open set by means of the flat output $\Vect{y}$ and its derivatives up to a finite order $q$:
\begin{subequations}\label{eq_FlatParam}
	\begin{align} 
	&\Vect{x} = \psi(\Vect{y},\Vect{\dot y},..., \Vect{y}^{(q-1)}),\\ 
	&\Vect{u} = \zeta(\Vect{y},\Vect{\dot y},..., \Vect{y}^{(q)})
	\end{align}
\end{subequations}
where  $(\psi, \zeta)$ are \textit{smooth functions} that give the trajectories of $\Vect{x}$ and $\Vect{u}$ as functions of the flat outputs and their time derivatives. The preceding expressions in \eqref{eq_FlatParam}, will be used to obtain the so called \textit{open-loop controls}.
The differential flatness found numerous applications, non-holonomic systems, among others (see \cite{sira2004} and the references therein). 

In the context of feedforwarding trajectories, the “degree of continuity” or the smoothness of the reference trajectory (or curve) is one of the most important factors. The smoothness of a trajectory is measured by the number of its continuous derivatives. We give the definitions on the trajectory continuity when it is represented by a parametric curve in the Appendix \ref{app:1-TrajConti}.

\section{Problem statement: Trajectory constraints fulfilment} \label{sec2ProblemStatement}

\subsection*{Notation}
Given the scalar
function $z\in C^{\kappa}(\rn,\rn)$ and the number $\alpha\in\nn$,
we denote by $\dv z{\alpha}$ the tuple of derivatives of $z$ up
to the order $\alpha\leqslant \kappa$: $\dv z{\alpha}=z,\,\dot{z},\ddot{z},\ldots,\,z^{(\alpha)}$.
Given the vector function $\Vect{v}=(v_{1},\ldots,v_{q})$, $v_{i}\in C^{\kappa}(\rn,\rn)$
and the tuple $\Vect{\alpha}=(\alpha_{1},\ldots,\alpha_{q})$, $\alpha_{i}\in\nn$,
we denote by $\dv v{\alpha}$ the tuple of derivatives of each component
$v_{i}$ of $\Vect{v}$ up to its respective order $\alpha_{i}\leqslant \kappa$:
$\dv v{\alpha}=v_{1},\ldots,v_{1}^{(\alpha_{1})},\,v_{2},\ldots,v_{2}^{(\alpha_{2})},\,\ldots\,,\,v_{q},\ldots,v_{q}^{(\alpha_{q})}$.

\subsection{General problem formulation}

Consider the nonlinear system
\begin{equation}
\Vect{\dot{x}}(t)=f(\Vect{x}(t),\Vect{u}(t))\label{eq:sys_nonlin}
\end{equation}
with state vector $\Vect{x}=(x_{1},\ldots,x_{n})$ and control input $\Vect{u}=(u_{1},\ldots,u_{m})$,
$x_{i},u_{j}\in C^{\kappa}([0,+\infty),\rn)$ for a suitable $\kappa\in\nn$.
We assume the state, the input and their derivatives to be subject
to both inequality and equality constraints of the form

\begin{subequations}
	\label{eq:constr}
	\begin{align}
	C_{i}(\dv x{\alpha_{i}^{x}} (t),\dv u{\alpha_{i}^{u}} (t))\leqslant 0 & \qquad\forall t\in[0,T],\ \forall i\in\{1,\ldots,\nu^{\mathrm{in}}\}\label{eq:constr_a}\\
	D_{j}(\dv x{\beta_{j}^{x}}(t),\dv u{\beta_{j}^{u}}(t))=0 & \qquad\forall t \in I_{j},\ \forall j\in\{1,\ldots,\nu^{\mathrm{eq}}\}\label{eq:constr_b}
	\end{align}
\end{subequations}
with each $I_{j}$ being either $[0,T]$ (continuous equality constraint)
or a discrete set $\{t_{1},\ldots,t_{\gamma}\}$, $0\leq t_{1}\leqslant \cdots\leqslant t_{\gamma}\leqslant T<+\infty$
(discrete equality constraint), and $\alpha_{i}^{x},\beta_{j}^{x}\in\nn^{n}$,
$\alpha_{i}^{u},\beta_{j}^{u}\in\nn^{m}$. We stress that the relations
(\ref{eq:constr}) specify objectives (and constraints) on the finite
interval $[0,T]$. Objectives can be also formulated as a concatenation
of sub-objectives on a union of sub-intervals, provided that some
continuity and/or regularity constraints are imposed on the boundaries
of each sub-interval. Here we focus on just one of such intervals. 

Our aim is to characterise the set of input and state trajectories
$(\Vect{x},\Vect{u})$ satisfying the  system’s equations (\ref{eq:sys_nonlin}) and the constraints
(\ref{eq:constr}). More formally we state the following problem.
\begin{problem}[Constrained trajectory set]
	\label{prob:ConstrProb} Let $\mathscr{C}$ be a subspace of $C^{\kappa}([0,+\infty),\rn)$.
	Constructively characterise the set $\mathscr{C}^{\mathrm{cons}}\subseteq\mathscr{C}^{n+m}$
	of all extended trajectories $(\Vect{x},\Vect{u})$ satisfying the system (\ref{eq:sys_nonlin})
	and the constraints (\ref{eq:constr}).
\end{problem}
	Problem \ref{prob:ConstrProb} can be considered as a generalisation
	of a constrained reachability problem (see for instance \cite{Faulwasser2014}). In such a reachability problem the stress is usually made on initial and final set-points and the goal is to find a suitable
	input to steer the state from the initial to the final point while possibly fulfilling the constraints. Here, we wish to give a functional
	characterisation of the overall set of extended trajectories $(\Vect{x},\Vect{u})$ 
	satisfying some given differential constraints. A classical constrained
	reachability problem can be cast in the present formalism by limiting
	the constraints $C_{i}$ and $D_{j}$ to $\Vect{x}$ and $\Vect{u}$ (and not their
	derivatives) and by forcing two of the equality constraints to coincide
	with the initial and final set-points.
	
	Problem \ref{prob:ConstrProb} is difficult to be addressed in its
	general setting. To simplify the problem, in the following we make
	some restrictions to the class of systems and to the functional space
	$\mathscr{C}$. As a first assumption we limit the analysis to differentially
	flat systems \cite{Fliess1995}.

\subsection{Constraints in the flat output space}

Let us assume that system (\ref{eq:sys_nonlin}) is differentially
flat with flat output\footnote{We recall that the flat output $\Vect{y}$ has the same dimension $m$ as the input vector $\Vect{u}$.}
\begin{equation}
\Vect{y}=(y_{1},\ldots,y_{m})=h(\Vect {x},\dv {u}{\rho^{u}})\,,\label{eq:flat_out}
\end{equation}
with $\rho^{u}\in\nn^{m}$. Following Equation \eqref{eq_FlatParam}, the parameterisation or the feedforwarding trajectories 
associated to the reference trajectory $\Vect{y_r}$ is

\begin{subequations}
	\label{eq:flat_param}
	\begin{align}
	\Vect{x}_r& =\psi(\dv {y_r}{\eta^{x}})\label{eq:flat_param_a}\\
	\Vect{u}_r & =\zeta(\dv {y_r}{\eta^{u}})\,,\label{eq:flat_param_b}
	\end{align}
\end{subequations}
with $\eta^{x}\in\nn^{n}$ and $\eta^{u}\in\nn^{m}$.

Through the first step of the dynamical extension algorithm \cite{Fliess1990}, we get the flat output dynamics
\begin{equation}
\left\{ \begin{aligned}y_{1}^{(k_{1})} & =\phi_{1}(\dv y{\mu_{1}^{y}},\dv u{\mu_{1}^{u}})\\
\vdots\\
y_{m}^{(k_{m})} & =\phi_{m}(\dv y{\mu_{m}^{y}},\dv u{\mu_{m}^{u}})\,,
\end{aligned}
\right.\label{eq:flat_out_dyn}
\end{equation}
with $\mu_{i}^{y}=(\mu_{i1}^{y},\ldots,\mu_{im}^{y})\in\nn^{m}$,
$\mu_{i}^{u}=(\mu_{i1}^{u},\ldots,\mu_{im}^{u})\in\nn^{m}$ and $k_{i}>\max_{j}\mu_{ji}^{y}$.
The original $n$-dimensional dynamics (\ref{eq:sys_nonlin}) and
the $K$-dimensional flat output dynamics (\ref{eq:flat_out_dyn})
($K=\sum_{i}k_{i}$) are in one-to-one correspondence through (\ref{eq:flat_out})
and (\ref{eq:flat_param}). Therefore, the constraints (\ref{eq:constr})
can be re-written as

\begin{subequations}
	\label{eq:constr_flat}
	\begin{align}
	\varGamma_{i}(\dv {y_r}{\omega_{i}^{\mathrm{in}}})\leqslant 0 & \qquad\forall t \in[0,T],\ \forall i\in\{1,\ldots,\nu^{\mathrm{in}}\}\label{eq:constr_flat_a}\\
	\varDelta_{j}(\dv {y_r}{\omega_{j}^{\mathrm{eq}}})=0 & \qquad\forall t \in I_{j},\ \forall j\in\{1,\ldots,\nu^{\mathrm{eq}}\}\label{eq:constr_flat_b}
	\end{align}
\end{subequations}
with $$\varGamma_{i}(\dv {y_r}{\omega_{i}^{\mathrm{in}}})=C_{i}(\dda{(\psi(\dv {y_r}{\eta_{^{x}}}))}{\alpha_{i}^{x}},\dda{\zeta(\dv {y_r}{\eta^{u}})}{\alpha_{i}^{u}}),$$
$$\varDelta_{j}(\dv {y_r}{\omega_{j}^{\mathrm{eq}}})=D_{j}(\dda{(\psi(\dv {y_r}{\eta^{x}})}{\beta_{j}^{x}},\dda{\zeta(\dv {y_r}{\eta^{u}})}{\beta_{j}^{u}})$$\\
and $\omega_{i}^{\mathrm{in}},\omega_{j}^{\mathrm{eq}}\in\nn^{m}$.

\begin{remark}
	We may use the same result to embed an \textit{input rate constraint} $\Vect{\dot u_r}$.
\end{remark}

Thus, Problem \ref{prob:ConstrProb} can be transformed in terms of
the flat output dynamics (\ref{eq:flat_out_dyn}) and the constraints
(\ref{eq:constr_flat}) as follows.
\begin{problem}[Constrained flat output set]\label{prob:OutConstrProb}\footnote{Here the max operator is applied elementwise on each vector.}
	Let $\mathscr{C}_{y}$ be a subspace of $C^{p}([0,+\infty),\rn)$
	with $p=\max((k_{1},\ldots,k_{m}),\omega_{1}^{\mathrm{in}},\ldots,\omega_{\nu^{\mathrm{in}}}^{\mathrm{in}},\omega_{1}^{\mathrm{eq}},\ldots,\omega_{\nu^{\mathrm{eq}}}^{\mathrm{eq}})$.
	Constructively characterise the set $\mathscr{C}_{y}^{\mathrm{cons}}\subseteq\mathscr{C}_{y}^{m}$
	of all flat outputs satisfying the dynamics (\ref{eq:flat_out_dyn})
	and the constraints (\ref{eq:constr_flat}).
\end{problem}
	Working with differentially flat systems allows us to translate, in
	a unified fashion, all the state and input constraints as constraints
	in the flat outputs and their derivatives (See (\ref{eq:constr_flat})).
	We remark that $\psi$ and $\zeta$ in (\ref{eq:flat_param}) are
	such that $\psi(\dv y{\eta_{x}})$ and $\zeta(\dv y{\eta_{u}})$ satisfy
	the dynamics of system (\ref{eq:sys_nonlin}) by construction. In
	other words, the extended trajectories $(\Vect{x},\Vect{u})$ of (\ref{eq:sys_nonlin})
	are in one-to-one correspondence with $\Vect{y}\in\mathscr{C}_{y}^{m}$ given
	by (\ref{eq:flat_out}). Hence, choosing $\Vect{y}$ solution of Problem
	\ref{prob:OutConstrProb} ensures that $\Vect{x}$ and $\Vect{u}$ given by (\ref{eq:flat_param})
	are solutions of Problem \ref{prob:ConstrProb}.

\subsection{Problem specialisation}

For any practical purpose, one has to choose the functional space
$\mathscr{C}_{y}$ to which all components of the flat output belong.
Instead of making reference to the space $\mathscr{C}^{\mathrm{gen}}\coloneqq C^{p}([0,+\infty),\rn)$,
mentioned in the statement of Problem \ref{prob:ConstrProb}, we focus
on the space $\mathscr{C}_{T}^{\mathrm{gen}}\coloneqq C^{p}([0,T],\rn)$.
Indeed, the constraints (\ref{eq:constr_flat}) specify finite-time
objectives (and constraints) on the interval $[0,T]$. Still, the problem
exhibits an infinite dimensional complexity, whose reduction leads to choose an approximation space $\mathscr{C}^{\mathrm{app}}$
that is dense in $\mathscr{C}_{T}^{\mathrm{gen}}$. A possible choice
is to work with parametric functions expressed in terms of basis functions
like, for instance, Bernstein-Bézier, Chebychev or Spline polynomials. 

A scalar Bézier curve of degree $N\in\nn$ in the Euclidean space
$\rn$ is defined as
\[
P(s)=\sum_{j=0}^{N}\alpha_{j}B_{jN}(s),\qquad s\in[0,1]
\]
where the $\alpha_{j}\in\rn$ are the control points and $B_{jN}(s)=\binom{N}{j}(1-s)^{N-j}s^{j}$
are Bernstein polynomials \cite{DeBoor}. For sake of simplicity,
we set here $T=1$ and we choose as functional space

\begin{equation}\label{eqB7}
\mathscr{C}^{\mathrm{app}} = \left\lbrace \sum_{0}^{N} \alpha_j B_jN  |  N \in \Naturals, (\alpha_j)_{0}^{N} \in \Reals^{N+1}, B_j \in \mathcal{C}^{0} ([0, 1], \Reals)\right\rbrace 
\end{equation} 

The set of Bézier functions of generic degree has the very useful
property of being closed with respect to addition, multiplication,
degree elevation, derivation and integration operations (see section \ref{SecTraj}).
As a consequence, any polynomial integro-differential operator applied
to a Bézier curve, still produces a Bézier curve (in general of different
degree). Therefore, if the flat outputs $\Vect{y}$ are chosen in $\mathscr{C}^{\mathrm{app}}$
and the operators $\varGamma_{i}(\cdot)$ and $\varDelta_{j}(\cdot)$
in (\ref{eq:constr_flat}) are integro-differential polynomials, then
such constraints can still be expressed in terms of Bézier curves
in $\mathscr{C}^{\mathrm{app}}$. We stress that, if some constraints
do not admit such a description, we can still approximate them up
to a prefixed precision $\varepsilon$ as function in $\mathscr{C}^{\mathrm{app}}$
by virtue of the denseness of $\mathscr{C}^{\mathrm{app}}$ in $\mathscr{C}_{1}^{\mathrm{gen}}$.
Hence we assume the following.
\begin{assumption}
Considering each flat output $y_r \in \mathscr{C}^{\mathrm{app}}$  defined as
$$ y_r =  \sum_{j=0}^{N} \alpha_j B_{jN}(s),$$
the constraints (\ref{eq:constr_flat}) can be written as
\begin{align} 
&\Gamma_i(\Vect{y_r}^{\left\langle \omega_{i}^{\mathrm{in}}\right\rangle}) = \sum_{k=0}^{N_i^{\mathrm{in}}} \lambda_{ik} B_{kN} (s), \label{Eq_ConstraintsA1} \\
&\Delta_j(\Vect{y_r}^{\left\langle \omega_{j}^{\mathrm{eq}}\right\rangle}) = \sum_{k=0}^{N_i^{\mathrm{eq}}} \delta_{jk} B_{kN} (s) \label{Eq_ConstraintsA2}
\end{align}

where
\begin{align*}
\lambda_{ik} = r_{ik}^{\mathrm{in}}(\alpha_0, \ldots, \alpha_N)\\
\delta_{jk} = r_{jk}^{\mathrm{eq}}(\alpha_0, \ldots, \alpha_N)\\
r_{ik}^{\mathrm{in}}, r_{jk}^{\mathrm{eq}} \in \Reals[\alpha_0, \ldots, \alpha_N]
\end{align*}
\ie the $\lambda_{ik}$ and $\delta_{jk}$ are polynomials in the $\alpha_0, \ldots, \alpha_N$.$\blacksquare$

\end{assumption}

Set the following expressions as$ \nu^{\mathrm{in}}$
\begin{align*}
&r^{\mathrm{in}} = (r_{1,0}^{\mathrm{in}}, \ldots, r^{\mathrm{in}}_{{ \nu^{\mathrm{in}}},{N}_{ \nu^{\mathrm{in}}}^{in}}), \\
&r^{\mathrm{eq}} = (r_{1,0}^{\mathrm{eq}}, \ldots, r^{\mathrm{eq}}_{{ \nu^{\mathrm{eq}}},{N}_{ \nu^{\mathrm{eq}}}^{eq}}), \\
&r = (r^{\mathrm{in}}, r^{\mathrm{eq}}),
\end{align*}
 the control point vector  $\Vect{\alpha} = (\alpha_1, \ldots, \alpha_N)$, and the basis function vector $\Vect{B} = (B_{1N}, \ldots,  B_{NN})$.
Therefore, we obtain a  semi-algebraic set defined as:
$$\mathscr{I}(r, \mathbb{A})= \left\lbrace \Vect{\alpha} \in \mathbb{A} \:| \: r^{\mathrm{in}} (\Vect{\alpha}) \leqslant 0, r^{\mathrm{eq}}(\Vect{\alpha}) = 0 \right \rbrace$$
for any parallelotope 
\begin{equation}\label{parallelo}
\mathbb{A} = [\munderbar{\alpha}_0, \bar{\alpha}_0] \times \cdots \times [\munderbar{\alpha}_N, \bar{\alpha}_N] , \munderbar{\alpha}_i, \bar{\alpha}_i \in \Reals \cup\lbrace -\infty, \infty \rbrace, \munderbar{\alpha}_i < \bar{\alpha}_i
\end{equation}

Thus $\mathscr{I}(r, \mathbb{A})$ is a semi-algebraic set associated to the constraints \eqref{eq:constr_flat}. The parallelotope $\mathbb{A}$ represents the trajectory sheaf of available trajectories, among which the user is allowed to choose a reference. The semi-algebraic set $\mathscr{I}(r, \mathbb{A})$ represents how the set $\mathbb{A}$ is transformed in such a way that the trajectories fulfill the constraints \eqref{eq:constr_flat}.
Then, picking an $\Vect{\alpha}$ in $\mathscr{I}(r, \mathbb{A})$ ensures that $y_r= \Vect{\alpha} \Vect{B}$ automatically satisfies the constraints \eqref{eq:constr_flat}.\\
The Problem \ref{prob:OutConstrProb} is then reformulated as :

\begin{problem}
For any fixed parallelotope $\mathbb{A}$, constructively characterise the semi-algebraic set $\mathscr{I}(r, \mathbb{A})$.
\end{problem}
This may be done through exact, symbolic techniques (such as, \eg the Cylidrical Algebraic Decomposition) or through approximation techniques yielding outer approximations $\mathscr{I}_l^{out}(r, \mathbb{A}) \supseteq \mathscr{I}(r, \mathbb{A}) $ and inner approximations $\mathscr{I}_l^{inn}(r, \mathbb{A}) \subseteq \mathscr{I}(r, \mathbb{A}) $ with $\lim\limits_{l \rightarrow \infty} \mathscr{I}_l^{out} = \lim\limits_{l \rightarrow \infty} \mathscr{I}_l^{inn} = \mathscr{I}$. $\blacksquare$ \\
This characterisation shall be useful to extract inner approximations of a special type yielding trajectory sheaves included in $\mathscr{I}(r, \mathbb{A})$. A specific example of this type of approximations will consist in disjoint unions of parallelotopes:

\begin{equation}
	\mathscr{I}_l^{inn} (r, \mathbb{A})= \bigcup \limits_{j \in I_l} \mathbb{B}_{l,j}, \quad \forall i, j \in I_l, \mathbb{B}_{l,i} \cap \mathbb{B}_{l,j} = \emptyset
\end{equation}

This class of inner approximation is of practical importance for end users, as the applications in Section \ref{sec5App} illustrate.

\subsection{Closed-loop trajectory tracking }
So far this chapter has focused on the design of open-loop trajectories while assuming that the system model is perfectly known and that the initial conditions are exactly known.
When the reference open-loop trajectories $(\Vect{x_r},\Vect{u_r})$ are well-designed \ie respecting the constraints and avoiding the singularities, as discussed above, the system is close to the reference trajectory. However, to cope with the environmental disturbances and/or small model uncertainties, the tracking of the constrained open-loop trajectories should be made robust using \textit{feedback control}. The feedback control guarantees the stability and a certain robustness of the approach, and is called the second degree of freedom of the primal controller (Stage B2 in figure \ref{fig_Scheme}). \\

We recall that some flat systems can be transformed via endogenous feedback and coordinate change to a linear dynamics \cite{Fliess1995, sira2004}. To make this chapter self-contained, we briefly discuss the \textit{closed-loop trajectory tracking} as presented in \cite{Martin2006}.

Consider a differentially flat system with flat output $\Vect{y} = (y_1, \ldots, y_m)$ ($m$ being the number of independent inputs of the system). Let $\Vect{y}_r (t) \in C^{\eta}(\Reals)$ be a reference trajectory for $\Vect{y}$. 
Suppose the desired open-loop state/ input trajectories $(x_r(t), u_r(t))$ are generated offline. We need now a feedback control to track them. 

Since the nominal open-loop control (or the feedforward input) linearizes the system, we can take a simple linear feedback, yielding the following closed-loop error dynamics:
\begin{equation}\label{eq_LinearControl}
\Vect{e}^{(\eta)} + \lambda_{\eta-1} \Vect{e}^{(\eta-1)} + \cdots + \lambda_1 \Vect{\dot e} + \lambda_0 \Vect{e} = 0
\end{equation}
where $\Vect{e} = \Vect{y} - \Vect{y}_r$ is the tracking error and the coefficients $\Lambda= [\lambda_0, \ldots, \lambda_{\eta-1}]$ are chosen to ensure an asymptotically stable behaviour (see e.g. \cite{Fliess1999}).
\begin{remark}Note that this is not true for all flat systems, in \cite{Hagenmeyer2002} can be found an example of flat system with nonlinear error dynamics.
\end{remark}

 Now let $(\Vect{x} , \Vect{u}) $ be the closed-loop trajectories of the system. These variables can be expressed in terms of the flat output $\Vect{y}$
as:
\begin{align}
\label{eq:expressionZInFlatOutput}
\Vect{x} = \psi(\dv {y}{\eta -1}), \quad \Vect{u} = \zeta (\dv {y}{\eta})
\end{align}

Then, the associated reference open-loop trajectories $(\Vect{x_r}, \Vect{u_r}) $ are given by
\begin{align*}
\Vect{x_r} &= \psi(\dv {y_r}{\eta-1}),  \quad \Vect{u_r} = \zeta (\dv {y_r}{\eta})
\end{align*}

Therefore, 
\begin{align*}
\Vect{x}=  \psi(\dv {y}{\eta-1})=  \psi(\dv {y_r}{\eta-1} + \dv {e}{\eta-1})
\end{align*}
and
\begin{align*}
\Vect{u}=  \zeta(\dv {y}{\eta})=  \zeta(\dv {y_r}{\eta} + \dv {e}{\eta}, -\Lambda \dv {e}{\eta}).
\end{align*}

As further demonstrated in \cite{Martin2006}[See Section 3.3], since the tracking error $\Vect{e}\rightarrow 0$ as $t \rightarrow\infty$ that means $\Vect{x}\rightarrow\Vect{x_r}$ and $\Vect{u}\rightarrow\Vect{u_r}$.

Besides the linear controller (Equation \eqref{eq_LinearControl}), many different linear and nonlinear feedback controls can be used to ensure convergence to zero of the tracking error. For instance, sliding mode control, high-gain control, passivity based control, model-free control, among others.

\begin{remark}
An alternative method to the feedback linearization, is the \textit{exact feedforward linearization} presented in \cite{Hagenmeyer2003a} where the problem of type "division by zero" in the control design is easily avoided. This control method removes the need for asymptotic observers since in its design the system states information is replaced by their corresponding reference trajectories. The robustness of the exact feedforwarding linearization was analyzed in \cite{Hagenmeyer2010}.
\end{remark}


\section{Preliminaries on Symbolic Bézier trajectory } \label{SecTraj}
To create a trajectory that passes through several points, we can use approximating or interpolating approaches. The interpolating trajectory that passes through the points is prone to oscillatory effects (more unstable), while the approximating trajectory like the Bézier curve or B-Spline curve is more convenient since it only approaches defined so-called \textit{control points} \cite{DeBoor} and have simple geometric interpretations. The Bézier/B-spline curve can be handled by conveniently handling the curve's control points. \\
The main reason in choosing the Bézier curves over the B-Splines curves, is the simplicity of their arithmetic operators presented further in this Section. Despite the nice local properties of the B-spline curve, the direct symbolic multiplication\footnote{The multiplication operator is essential when we want to work with polynomial systems.} of B-splines lacks clarity and has partly known practical implementation \cite{Morken1991}.  \\

In the following Section, we start by presenting the \textit{Bézier curve} and its properties. Bézier curves are chosen to construct the reference trajectories because of their nice properties (smoothness, strong convex hull property, derivative property, arithmetic operations). They have their own type basis function, known as the Bernstein basis, which establishes a relationship with the so-called control polygon. A complete discussion about Bézier curves can be found in \cite{Prautzsch2002}. 
Here, some basic and key properties are recalled as a preliminary knowledge.

\subsection{Definition of the Bézier curve}
A Bézier curve is a parametric one that uses the Bernstein polynomials as a basis. An $n$th degree Bézier curve is defined by
\begin{equation}
f(t)=  \sum\limits_{j=0}^N c_j B_{j,N}(t), \quad 0 \leqslant t \leqslant 1  \\
\end{equation}
where the $c_j$ are \textit{the control points} and the basis functions $B_{j,N}(t)$ are the \textit{Bernstein polynomials} (see Figure \ref{fig_BernsteinBasis}). The $B_{j,N}(t)$ can be obtained explicitly by:
\begin{equation*}
B_{j,N}(t)= \binom{N}{j}(1-t)^{N-j} t^j \text{ for } j = 0, \ldots, N.
\end{equation*} 

or by recursion with the De Casteljau formula:
\begin{equation*}
B_{j,N}(t) = (1-t) B_{j,N-1}(t) + t B_{j-1,N-1}(t).
\end{equation*}
\begin{figure}
	\centering
	\includegraphics[width=4in]{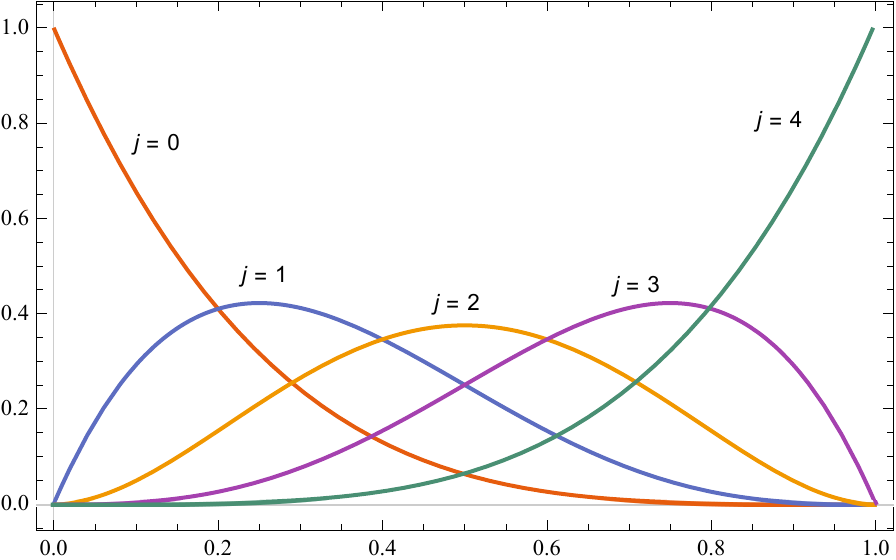}
	\caption{Bernstein Basis for degree $N=4$.}
	\label{fig_BernsteinBasis}
\end{figure}

\begin{figure}
	\centering
	\includegraphics[width=3in]{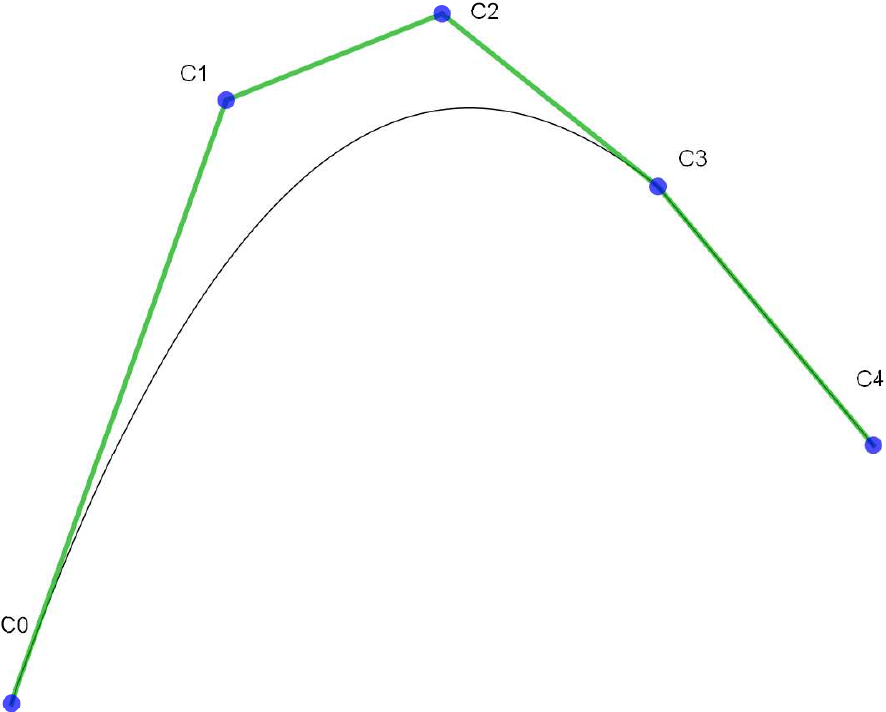}
	\caption{The convex hull property for Bézier curve  ($N=4$) with control points $c_j$($j = 0, \ldots, 4$).}
	\label{fig_BezierCurve}
\end{figure}

\subsection{Bézier properties}\label{subSecBezProp}
For the sake of completeness, we here list some important Bézier-Bernstein properties. 
\begin{lemma}
	Let $n$ be a non-negative polynomial degree. The Bernstein functions have the following properties:
	\begin{enumerate}
		\item \textit{Partition of unity.} $\sum\limits_{j=0}^n B_{j,N}(t) \equiv 1$\\
		This property ensures that the relationship
		between the curve and its defining Bézier points is invariant under affine transformations.
		
		\item \textit{Positivity.} If $t \in [0,1]$ then $B_{j,N}(t)>0$.\\
		It guarantees that the curve segment lies completely within the convex hull of the control points (see Figure \ref{fig_BezierCurve}). 
		
		\item \textit{Tangent property.} For the start and end point, this guarantees $f(0) = c_0$ and $f(1) = c_N$ but the curve never passes through the intermediate control points.
		
		\item \textit{Smoothness.} $B_{j,N}(t)$ is $N-1$  times continuously differentiable. Hence,
		increasing degree increases regularity.
	\end{enumerate}
\end{lemma}

\subsection{Quantitative envelopes for the  Bézier curve } \label{subSecBezierPolygon}
Working with the Bézier curve control points in place of the curve itself allows a simpler explicit representation. However, since our framework is not based on the Bézier curve itself, we are interested in the localisation of the Bézier curve with respect to its control points, \ie the control polygon. In this part, we review a result on \textit{sharp quantitative bounds} between the Bézier curve and its control polygon \cite{Nairn1999,DavidLu1999}. For instance, in the case of a quadrotor (discussed in Section \ref{sec2Quad}), once we have selected the control points for the reference trajectory, these envelopes describe the exact localisation of the quadrotor trajectory and its distance from the obstacles. These quantitative envelopes may be of particular interest when avoiding corners of obstacles which traditionally in the literature \cite{Radmanesh2016} are modelled as additional constraints or introducing safety margin around the obstacle.

We start by giving the definition for the control polygon.
\begin{definition}\label{DefControlPolygon}
	(Control polygon for Bézier curves (see \cite{Nairn1999})). Let $f=\sum_{j=0}^{N}c_j B_{j,N} (t) $ be a scalar-valued Bézier curve. \textit{The control polygon} $\Gamma_f=\sum_{j=0}^{N}c_j H_{j} (t)$ of $f$ is a piecewise linear function connecting the points with coordinates $(t_j^*, c_j)$ for $j=0, \ldots, N$ where the first components $t_j^* = \frac{j}{N}$ are the Greville abscissae. The hat functions $H_j$ are piecewise linear functions defined as: 
	$$H_j(t) = \begin{cases}
&\frac{t-t^*_{j-1}}{t^*_{j} -t^*_{j-1}} \quad t \in [t^*_{j-1}, t^*_j]\\
 &\frac{t^*_{j+1} -t}{t^*_{j+1} -t^*_{j}} \quad t \in [t^*_{j}, t^*_{j+1}]\\
 &0 \quad \text{otherwise}.
	\end{cases}$$
\end{definition}

An important detail is the \textit{maximal distance} between a Bézier segment and its control polygon. For that purpose, we recall a result from \cite{Nairn1999}, where sharp quantitative bounds of control polygon distance to the Bézier curve are given.

\begin{theorem} \label{ThConvergenceControlPolygon}
	(See \cite{Nairn1999}, Theorem 3.1) Let $f =\sum_{j=0}^{N}c_j B_{j,N} $ be a scalar Bézier curve and let $\Gamma_f$ be its control polygon. Then the maximal distance from $f$ to its control polygon is bounded as:
	\begin{equation}
	\norm{f  - \Gamma_f}_{\infty,\left[0, 1 \right] }\leqslant \mu_{\infty}(N) \norm{\Delta_2 c}_{\infty} =D_{\max}
	\end{equation}
	where the constant $\mu_{\infty}(N)=\dfrac{\lfloor N/2\rfloor \lceil N/2\rceil }{2N}$ 	
\footnote{	Note that the notation$\lceil x \rceil$ means the ceiling of $x$, i.e. the smallest integer greater than or equal to $x$, and the notation$\lfloor x \rfloor$ means the floor of $x$, i.e. the largest integer less than or equal to $x$.}	only depends on the degree $N$ and the second difference of the control points $\norm{\Delta_2 c}_{\infty } := \max_{0 < j < N} \vert \Delta_2 c_j \vert$. 
\end{theorem}
The $j^\th$ second difference of the control point sequence $c_j$ for $j=0, \ldots, N$ is given by:
\begin{equation*}
\Delta_2 c_j = c_{j-1} - 2 c_j + c_{j+1}.
\end{equation*}

Based on this maximal distance, \textit{Bézier curve's envelopes} are defined as two piecewise linear functions: 
\begin{itemize}
\item the \textit{lower envelope} $\underline{ \Gamma}_f = \sum_{j=0}^{N}\underline{e}_j H_{j} =  \sum_{j=0}^{N} (c_j-D_{\max}) H_{j}$ and,
\item the \textit{upper envelope} $\bar{\Gamma}_f =\sum_{j=0}^{N}\bar{e}_j H_{j} =  \sum_{j=0}^{N} (c_j+D_{\max}) H_{j}$  
\end{itemize}
such that $\underline{ \Gamma}_f\leqslant f  \leqslant  \bar{\Gamma}_f$. \\
 The envelopes are improved by taking $\underline{e}_0 = \bar{e}_0 = c_0$ and $\underline{e}_N = \bar{e}_N = c_N$ and then \textit{clipped} with the standard Min-Max bounds \footnote{Unfortunately the simple Min-Max bounds define very large envelopes when applied solely.}. The Min-Max bounds yield rectangular envelopes that are defined as
\begin{definition}
(Min-Max Bounding box (see \cite{Prautzsch2002})). Let $f=\sum_{j=0}^{N}c_j B_{j,N} $ be a Bézier curve. As a consequence of the convex-hull property, \textit{a min-max bounding box} is defined for the Bézier curve $f$  as:
$$  \min_{0 < j < N} c_j \leqslant\sum_{j=0}^{N}c_j B_{j,N}  \leqslant  \max_{0 < j < N} c_j.$$
\end{definition}

\begin{remark}
	As we notice, the maximal distance between a Bézier segment and its control polygon is bounded in terms of	the second difference of the control point sequence and a constant that depends only on the degree of the polynomial. Thus, by elevating the degree of the Bézier control polygon, \ie the subdivision (without modifying the Bézier curve), we can arbitrary reduce the distance between the curve and its control polygon.
\end{remark}

\subsection{Symbolic Bézier operations} \label{subSecBezierOperation}
In this section, we present the Bézier operators needed to find the Bézier control points of the states and the inputs.
Let the two polynomials $f(t)$ (of degree $m$) and $g(t)$ (of degree $n$) with \textit{control points} $f_j$ and $g_j$ be defined as follows:

\begin{align*}
f(t)=  \sum\limits_{j=0}^m f_j B_{j,m}(t), \quad 0 \leqslant t \leqslant 1  \\
g(t)=  \sum\limits_{j=0}^n g_j B_{j,n}(t), \quad 0 \leqslant t \leqslant 1
\end{align*}

We now show how to determine the control points for the degree elevation and for the arithmetic operations (the sum, difference, and product of these polynomials). For further information on Bézier operations, see \cite{Farouki1988}. Some illustrations of the geometrical significance of these operations are included in the Appendix \ref{app:1-BezierGeomSign}.

\paragraph{Degree elevation:}
To increase the degree from $n$ to $n+r$ and the number of control points from $n+1$ to $n+r+1$ without changing the shape, the new control points $b_j$ of the $(n+r)$th Bézier curve are given by:
\begin{equation}\label{eqDegElev}
b_j = \sum_{i=\max(0,j-r)}^{\min(n,j)} \dfrac{\binom{n}{i} \binom{r}{j-i}}{\binom{n+r}{j}} g_i \quad \quad j = 0,1, \ldots, n+r
\end{equation}
The latter constitutes the so-called \textit{augmented control polygon}. The new control points are obtained as convex combinations of the original control points. This is an important operation exploited in addition/subtraction of two control polygons of different lengths and in approaching the curve to a new control polygon by refining the original one.

\paragraph{Addition and subtraction:} 
If $m=n$ we simply add or subtract the coefficients
\begin{equation}
f(t) \pm g(t) = \sum\limits_{j=0}^m (f_j \pm g_j) B_{j,m}(t) \\
\end{equation}
If $m>n$, we need to first elevate the degree of $g(t)$ $m-n$ times using \eqref{eqDegElev} and then add or subtract the coefficients.

\paragraph{Multiplication:}
Multiplication of two polynomials of degree $m$ and $n$ yields a degree $m+n$ polynomial
\begin{equation}\label{eqBezierMultip}
f(t)g(t) = \sum_{j=0}^{m+n} \underbrace{ \left(  \sum_{i=\max(0, j-n)}^{\min(m,j)} \dfrac{\binom{m}{i} \binom{n}{j-i}}{\binom{m+n}{j}} f_i g_{j-i} \right) }_{\text{Control points of the product}}B_{j, m+n} (t)
\end{equation}

\subsection{Bézier time derivatives}

We give the derivative property of the Bézier curve  in Proposition \ref{Pro1} which is crucial in establishing the constrained trajectory procedure.

\begin{lemma} (see \cite{Lyche2002})
	The derivative of the $j$th Bernstein function of degree $n \geqslant 1$ is given by
	\begin{equation}
	DB_{j,N}(t) = N \left( B_{j-1,N-1}(t) - B_{j,N-1}(t)  \right) \text{ for } j = 0, \ldots, N. 
	\end{equation} 
	for any real number $t$ and where $B_{-1,N-1} =B_{N,N-1} =0$.\\
\end{lemma}

\begin{proposition}\label{Pro1}
	If the flat output or the reference trajectory $y$ is a Bézier curve, its derivative is still a Bézier curve and we have an explicit expression for its control points.
\end{proposition}
\begin{proof}
	Let $y^{(q)}(t)$ denote the $q$th derivative of the flat output $y(t)$. 
	We use the fixed time interval $T =t_f -t_0$ to define the time as $t = T \tau,$  $0\leqslant \tau \leqslant 1$.
	We can obtain $y^{(q)}(\tau)$ by computing the $q$th derivatives of the Bernstein functions. 
	\begin{equation}
	y^{(q)}(\tau)= \dfrac{1}{T^{q}} \sum_{j=0}^{N} c_j B_{j,N}^{(q)}(\tau)  
	\end{equation}
	
	Letting $c_j^{(0)}= c_j$, we write
	\begin{equation}
	y(\tau) = y^{(0)}(\tau) = \sum_{j=0}^{N} c_j^{(0)} B_{j,N}(\tau)  
	\end{equation}
	
	Then,
	\begin{equation}\label{eqBézierDeriv}
	y^{(q)}(\tau)  = \sum_{j=0}^{N-q} c_j^{(q)} B_{j,N-q}(\tau)  
	\end{equation}
	
	with derivative control points such that
	\begin{align}\label{eqControlPointsDerivative}
	c_j^{(q)} = \begin{cases}
	c_j,\quad  
	&q=0\\[.5ex]
	\: \dfrac{(N-q+1)}{T^{q}}\left( c_{j+1}^{(q-1)} - c_j^{(q-1)}\right) , &q>0. 
	\end{cases}
	\end{align}\\
\end{proof}

We can deduce the explicit expressions for all lower order derivatives up to order $N-1$.
This means that if the reference trajectory $y_r(t)$ is a Bézier curve of degree $N > q $ ($q$ is the derivation order of the flat output $y$), by differentiating it, all states and inputs are given in straightforward Bézier form.

\begin{example}

 Through a simple example of a double integrator, we want to represent the link between the time interval and the time derivatives. For a changing position $y$, its time derivative $\dot y$ is its velocity, and its second derivative with respect to time $\ddot y$, is its acceleration. Even higher derivatives are sometimes also used: the third derivative of position with respect to time is known as the jerk. 

We here want to show \textit{the effect of the fixed time period }$T$ on the velocity, acceleration, etc. We remark the connection between the time scaling parameter appearing in the trajectory parameterization. We have a simple double integrator defined as:
\begin{equation}
\ddot y = u
\end{equation}

As a reference trajectory, we choose a Bézier curve $y = \sum_{i=0}^N a_i B_{i,N}$ of order $N=4$.  Due to the Bézier derivative property, we can explicitly provide the link between the time interval $T$ and control points of the Bézier curve's derivatives.
\begin{subequations}
	\begin{align}
	\dot y = \sum_{i=0}^{N-1}  a_i^{(1)} B_{i,N-1} \\
	\ddot y = \sum_{i=0}^{N-2} a_i^{(2)} B_{i,N-2} 
	\end{align}
\end{subequations}
where $a_i^{(1)}$ and $a_i^{(2)}$ are the control points of the first and the second derivative of the B-spline curve respectively. 
We have the expressions of the $a_i^{(1)}$ and $a_i^{(2)}$  in terms of the $a_i$. This fact allow us to survey when the desired reference trajectory will respect the input constraints i.e. $a_i^{(2)} = f_1(a_i^{(1)} ) = f_2( a_i)$. That means that if $ \forall a_i^{(2)} <K$ then $u<K$.

\begin{figure}
	\begin{center}\hspace*{-15ex}
		\includegraphics[width=7in]{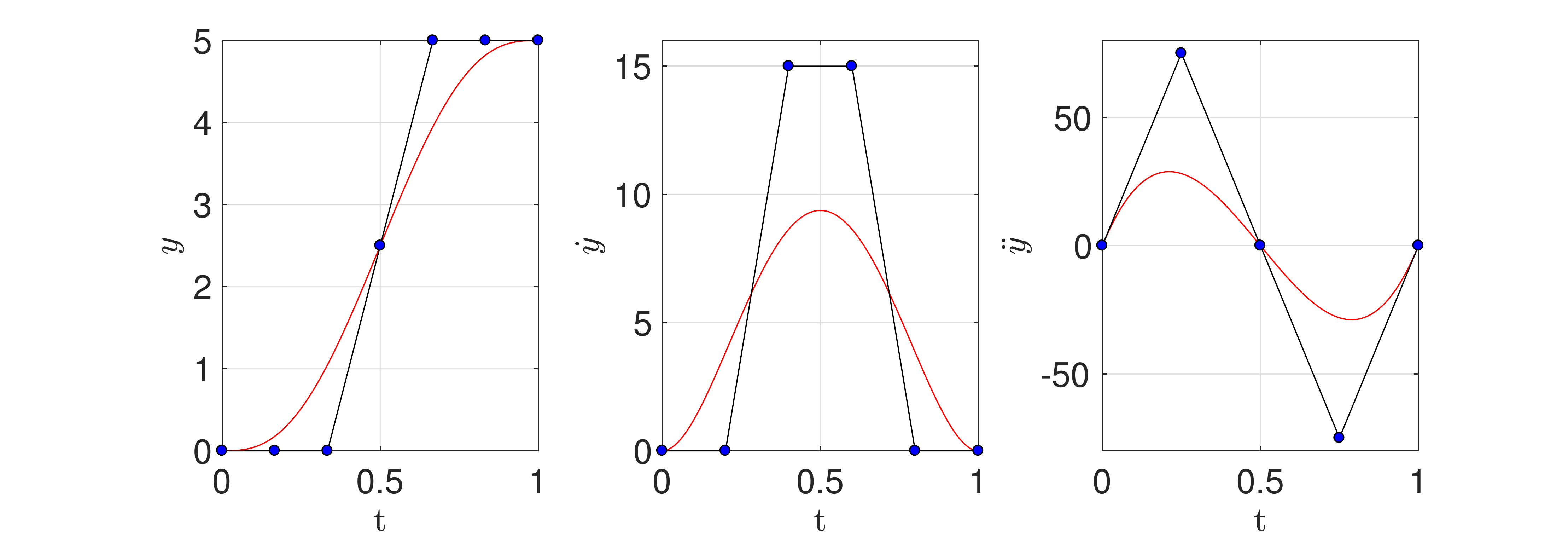}
		\caption{The time derivatives when $T=1$ }\label{figTimeDer1}
	\end{center}
\end{figure} 
\begin{figure}
	\begin{center}\hspace*{-15ex}
		\includegraphics[width=7in]{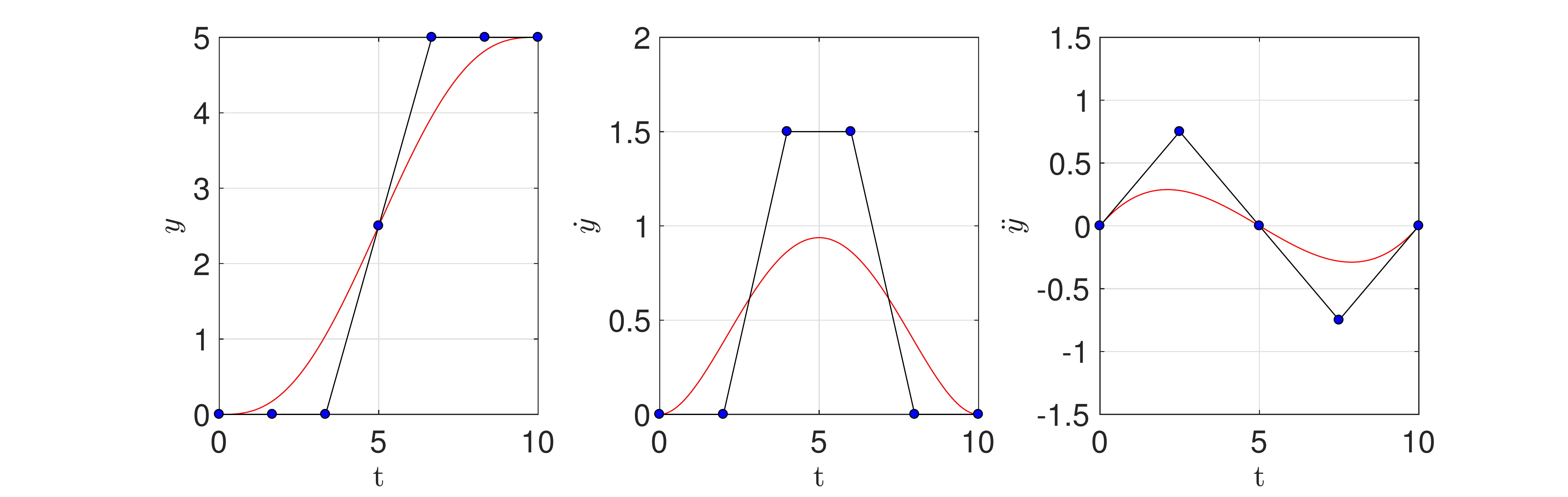}
		\caption{The time derivatives when $T=10$ }\label{figTimeDer10}
	\end{center}
\end{figure} 
\begin{figure}
	\begin{center}\hspace*{-15ex}
		\includegraphics[width=7in]{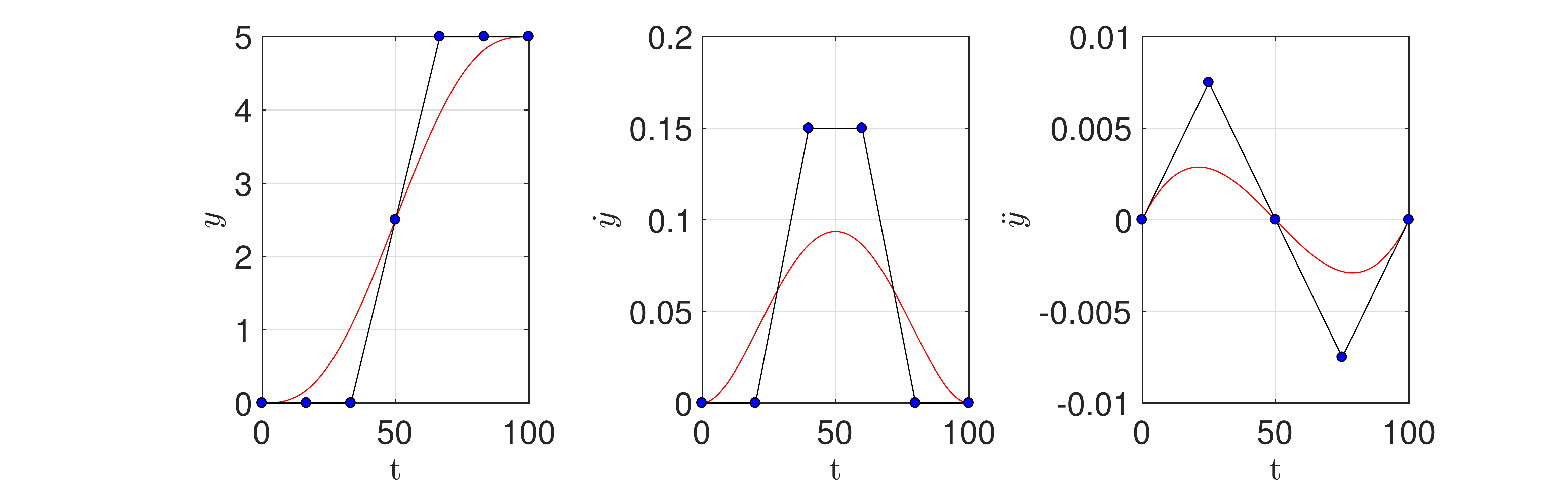}
		\caption{The time derivatives when $T=100$ }\label{figTimeDer100}
	\end{center}
\end{figure} 

\end{example}

\begin{proposition}\label{thBézierFlatness}
	If we take a Bézier curve as reference trajectory $y_r(t)=  \sum\limits_{j=0}^N c_j B_{j,N}(t)$  for a flat system such that the input is a polynomial function of the flat output and its derivatives, then the open loop input is also a Bézier curve
	$u_r = B(y_r,..., y_r^{(q)}) = \sum\limits_{i=0}^m U_i B_{i,m}(t) $.\\
\end{proposition}

\begin{remark}
	We should take a Bézier curve of degree $N > q$ to avoid introducing discontinuities in the control input. 
\end{remark}

\begin{example}
	In the case of a chain of integrators $u_r(t)= y_r^{(q)}(t) $ by imposing for all $K_l \leqslant c_j^{(q)} \leqslant K_h$, we ensure an input constraint $K_l \leqslant u_r(t) \leqslant K_h$.
\end{example}



\section{Constrained feedforward trajectory procedure}\label{SecProcedure}

We aim to find a feasible Bézier trajectory (or a set of feasible trajectories, and then make a suitable choice) $\Vect{y_r}(t)$ between the initial conditions $\Vect{y_r}(t_0) = \Vect{y}_{\text{initial}}$ and the final conditions  $\Vect{y_r}(t_f) = \Vect{y}_{\text{final}}$. We here show the procedure to obtain the \textit{Bézier control points} for the constrained nominal trajectories $(\Vect{y_r}, \Vect{x_r} ,\Vect{u_r})$. \\

Given a \textit{differentially flat system} $ \Vect{\dot x} = f(\Vect{x}, \Vect{u})$, the reference design procedure can be summarized as:
\begin{enumerate}
	\item 
	 Assign to each flat output (trajectory) $y_i$ a \textit{symbolic Bézier curve}   $y_r(t)=  \sum\limits_{j=0}^N \alpha_j B_{j,N}(t)$ of a suitable degree $N > q$ ($q$ is the time derivatives of the flat output) and where  $\Vect{\alpha} = (\alpha_0, \ldots, \alpha_N) \in \Reals^{N+1}$ are its control points.
	\item Compute the needed derivatives of the flat outputs using Equation \eqref{eqBézierDeriv}.
	\item Use the Bézier operations to produce the system model relationships \eqref{Eq_ConstraintsA1}-\eqref{Eq_ConstraintsA2}, and to find the \textit{state reference Bézier curve} $\Vect{x_r}(t)=  \sum\limits_{i=0}^m X_i B_{i,m}(t)$ and  \textit{input reference Bézier curve} $\Vect{u_r}(t)=  \sum\limits_{j=0}^m U_j B_{j,m}(t)$ respectively, such that $ (X_i, U_j)= r_{k} (\alpha_0, \ldots, \alpha_N), k=0, \ldots,m+n+2 )$ are functions of the \textit{output control points}.
	
	\item If needed, calculate the corresponding \textit{augmented control polygons} by elevating the degree of the original control polygons in order to be closer to the Bézier trajectory.
	\item Specify the initial conditions, final conditions, or intermediate conditions on the flat output or on any derivative of the flat output that represent a direct equality constraint on the Bézier control points. Each flat output trajectory has its control points fixed as follows:
	
	\begin{subequations}\label{eqOutputConstraints}
		\begin{align}
		\alpha_0^{(i)} &= y^{(i)}(t_0), \\
		\alpha_N^{(i)} &= y^{(i)}(t_f),  \text{  for } i = 0, \ldots\, q,\\
		\alpha_j &\in [\munderbar{\alpha_j}, \bar \alpha_j] \text{  for } j= 1, \ldots\, N-1,
		\end{align}
	\end{subequations}
	where  $\munderbar{\alpha_j}, \bar \alpha_j \in \Reals$ are the limits of the $j^\th$ control point. 
	By using the Bézier properties, we will construct a set of constraints by means of its control points.
	We have a special case for the paralellotope where the first and last control point are fixed  $\munderbar{\alpha_0} =  \bar \alpha_0= y(t_0)$  and $\munderbar{\alpha_N} =  \bar \alpha_N= y(t_f)$ respectively.

	\item We consider a constrained method based on the Bézier control points since the control point polygon captures important geometric properties of the Bézier curve shape. The conditions on the output Bézier control points $\alpha_j$, the state Bézier control points $X_i$ and the the input control points $U_j$ result in a semi-algebraic set (system of polynomial equations and/or inequalities) defined as:

	\begin{equation} \label{eqSASys}
	\mathscr{I}(r, \mathbb{A})=  \left\lbrace \Vect{ \alpha} \in   \mathbb{A} \:| \:	r_k(\Vect{ \alpha}) *_k 0 , k \in \left\lbrace 1, \ldots, l\right\rbrace , *_k \in  \left\lbrace <, \leqslant, >, \geqslant, = , \neq \right\rbrace  \right \rbrace 
	\end{equation}
	Depending on the studied system, the output constraints can be defined as in equation \eqref{parallelo}, or remain as $\mathbb{A} = \Reals^{N+1}$.
	\item Find the regions of the \textit{control points} $\alpha_j$, $j= 1, \ldots\, N-1,$ solving the system of equality/inequalities \eqref{eqSASys} by using an appropriate method. We present two kind of possible methods in Section \ref{sec4Feasibility}.
\end{enumerate}

\section{Feasible control points regions\label{sec4Feasibility}}

Once we transform all the system trajectories through the symbolic Bézier flat output, the problem is formulated as a system of functions (equations and inequalities) with Bézier control points as parameters (see equation \eqref{eqSASys}). Consequently the following question raises:
\begin{question}
	  How to find the regions in the space of the parameters (Bézier control points) where the system of functions remains valid \ie the constrained set of feasible feed-forwarding trajectories?
\end{question}

This section has the purpose to answer the latter question by reviewing two methods from \textit{semialgebraic geometry} \footnote{The theory that studies the real-number solutions to algebraic inequalities with-real number coefficients, and mappings between them, is called semialgebraic geometry.} :\\

In \textit{the first method}, we formulate the regions for the reference trajectory control points search as a \textit{Quantifier Elimination (QE)} problem. The QE is a powerful procedure to compute an \textit{equivalent} quantifier-free formula for a given first-order formula over the reals \cite{Tarski1998, Coste2002}. Here we briefly introduce the QE method.\\
Let  $f_i(X, U) \in \Rationals [X,U], i = 1, \ldots, l$ be polynomials with rational coefficients where:
\begin{itemize}
	\item $X = (x_1, \ldots, x_n) \in \Reals^n$ is a vector of quantified variables
	\item $U= (u_1, \ldots, u_m) \in \Reals^m$ is a vector of unquantified (free) variables.
\end{itemize}
The \textit{quantifier-free} Boolean formula $\varphi(X, U)$ is a combined expression of polynomial equations ($f_i(X, U)= 0$) , inequalities  ($f_i(X, U)\leq 0$), inequations ($f_i(X, U)\neq 0$) and strict inequalities  ($f_i(X, U) > 0$) that employs the logic operators $\land$ (and), $\lor$ (or), $\Rightarrow$ (implies) or $\Leftrightarrow $ (equivalence). \\
A prenex or \textit{first-order formula} is defined as  follows:
$$G(X, U) = (Q_1 x_1) \ldots (Q_n x_n) [\varphi(X, U)]$$
where $Q_i$ is one of the quantifiers $\forall$(for all) and $\exists$ (there exists). 
Following the Tarski Seidenberg theorem (see \cite{Coste2002}), for every prenex formula $G(X,U)$ there exists an equivalent quantifier-free formula $\psi(U)$ defined by the free variables. \\

 The goal of the QE procedure is to compute an equivalent quantifier free formula $\psi(U)$ for a given first-order formula.  It finds the feasible regions of free variables $U$ represented as semialgebraic set where $G(X,U)$ is true. If the set $U$ is non-empty, there exists a point $u \in \Reals^m$ which simultaneously satisfies all of the equations/inequalities. Such a point is called a feasible point and the set $U$ is then called feasible. If  the set $U$ is empty, it is called unfeasible. In the case when $m=0$,  $i.e$. when all variables are quantified, the QE procedure decides whether the given formula is true or false (decision problem). 
For instance, 
\begin{itemize}
	\item given a first order formula $\forall x\,[x^2 + bx+c>0]$, the QE algorithm gives the equivalent quantifier free formula $b-4c<0$;
	\item given a first order formula $\exists x\,[ax^2 + bx+c=0]$, the QE algorithm gives the equivalent quantifier free formula $(a\neq 0 \land b^2 -4ac \geq 0) \lor (a=0 \land b \neq 0) \lor (a=0 \land b=0 \land c=0).$
\end{itemize}

As we can notice, the quantifier free formulas represent the semi-algebraic sets (the conditions) for the unquantified free variables verifying the first order formula is \textit{true}. Moreover,  given an input formula without quantifiers, the QE algorithm produces a \textit{simplified formula}. For instance  (for more examples, see \cite{Brown2003}),
\begin{itemize}
	\item given an input formula $(ab \leqslant 0) \land (a+b = 0) \land (b^2+a^2>0 )\lor (a^2 = -b^2)$, the QE algorithm gives the equivalent simplified formula $a+b=0.$ 
\end{itemize}
On the other hand, given an input formula without unquantified free variables (usually called closed formula) is either \textit{true} or \textit{false}.\\

The symbolic computation of the \textit{Cylindrical Algebraic Decomposition} (CAD) introduced by Collins \cite{Collins1975} is the best currently known QE algorithm for solving real algebraic constraints (in particular parametric and \textit{non-convex} case) (see \cite{Strzebonski2006}). This method gives us \textit{an exact solution}, a simplified formula describing the semi-algebraic set. 

The QE methods, particularly the CAD, have already been used in various aspects of control theory (see \cite{Ratschan2012, Anai2014} and the references therein): robust control design, finding the feasible regions of a PID controller, the Hurwitz and Schur stability regions, reachability analysis of nonlinear systems,  trajectory generation \cite{Lindemann2006}.

\begin{remark} (On the complexity)
	Unfortunately the above method rapidly becomes slow due to its double exponential complexity \cite{Magron2015}. Its efficiency strongly depends on the number and on the complexity of the variables (control points) used for a given problem. The computational complexity of the CAD is double exponential \ie bounded by $(sd)^{2 \mathcal{O}(n)}$ for a finite set of $s$ polynomials in $n$ variables, of degree $d$. There are more computationally efficient QE methods than the CAD, like the \textit{Critical Point Method} \cite{Basu2014} (it has single exponential complexity in $n$ the number of variables) and the cylindrical algebraic sub-decompositions \cite{Wilson} but to the author knowledge there are no available implementations.
	
\end{remark}

For more complex systems, the exact or symbolic methods are too computationally expensive. There exist methods that are numerical rather than exact.\\

\textit{As a second alternative method}, we review one such method based on approximation of the exact set with more reasonable computational cost. The second method known as the Polynomial Superlevel Set (PSS) method, based on the paper \cite{Dabbene2017} instead of giving us exact solutions tries to approximate the set of solutions by minimizing the $L^1$ norm of the polynomial. It can deal with more complex problems.\\

\subsection{Cylindrical Algebraic Decomposition} \label{SubsecCAD}

In this section, we give a simple introduction to the Cylindrical Algebraic Decomposition.
\paragraph{Input of CAD:} 
As an input of the CAD algorithm, we define a set of polynomial equations and/or inequations in $n$ unknown symbolic variables (in our case, the control points) defined over real interval domains. 

\paragraph{Definition of the CAD:}
The idea is to develop a sequence of projections that drops the dimension of the semi-algebraic set by one each time.
Given a set $S$ of polynomials in $R^n$, a \textit{cylindrical algebraic decomposition} is a decomposition of $R^n$ into finitely many connected semialgebraic sets called \textit{cells}, on which each polynomial has constant sign, either $+$, $-$ or $0$. To be cylindrical, this decomposition must satisfy the following condition: If $1 \leqslant k < n$ and $\pi$ is the projection from $R^n$ onto $R^{n-k}$ consisting in removing the $k$ last coordinates, then for every pair of cells $c$ and $d$, one has either $\pi(c) = \pi(d)$ or $\pi(c) \cap \pi(d) = \emptyset$. This implies that the images by $\pi$ of the cells define a cylindrical decomposition of $R^{n-k}$.

\paragraph{Output of CAD:}
As an output of this symbolic method, we obtain the total algebraic expressions that represent \textit{an equivalent simpler form} of our system. Ideally, we would like to obtain a parametrization of all the control points regions as a closed form solution. Finally, in the case where closed forms are computable for the solution of a problem, one advantage is to be able to overcome any optimization algorithm to solve the problem for a set of given parameters (numerical values), since only an evaluation of the closed form is then necessary.\\

The execution runtime and memory requirements of this method depend of the dimension of the problem to be solved because of the computational complexity. For the implementation part, we will use its \textit{Mathematica}  implementation\footnote{ see \url{https://reference.wolfram.com/language/ref/CylindricalDecomposition.html}} (developed by Adam Strzebonski). Other implementations of CAD are QEPCAD, Redlog, SyNRAC, Maple.

\begin{example}
From \cite{Kauers2011}, we present an example in which we want to find the regions of the parameters $(a,b) \in \Reals^2$  \textit{where} the following formula is true, not only answering if the formula is true or not.\\
Having as input 

	\begin{equation*}
	F = \left\lbrace (a,b) \in \Reals^2 : f_1 (a,b)= \sqrt{a^2 - b^2} + \sqrt{ab-b^2} -a>0, \quad f_2(a,b) =0<b<a \right\rbrace                                                  
	\end{equation*}

the corresponding CAD output is given by
	\begin{equation*}
	\left\lbrace  a>0 \land b< \dfrac{4}{5}a \right\rbrace 
	\end{equation*}

As we notice, given a system of equations and inequalities formed by the control points relationship as an input, the CAD returns a simpler system that is equivalent over the reals.

\end{example}

\subsection{Approximations of Semialgebraic Sets}
Here we present a method based on the paper \cite{Dabbene2017} that tries to approximate the set of solutions. Given a set $$\mathcal{K} = \lbrace x \in \Reals^n : g_i(x)\geqslant 0, i = 1, 2, \ldots, m\rbrace$$ which is compact, with non-empty interior and described by given real multivariable polynomials $g_i(x)$ and a compact set $\mathcal{B}  \supset \mathcal{K} $, we aim at determining a so-called \textit{polynomial superlevel set} (PSS)
$$U(p) =\lbrace x \in \mathcal{B} : p(x) \geqslant 1\rbrace$$
The set $\mathcal{B} $ is assumed to be an $n$-dimensional hyperrectangle.
The PSS can capture the main characteristics of $\mathcal{K} $ (it can be non convex and non connected) while having at the same time a simpler description than the original set. It consists in finding a polynomial $p$ of degree $d$ whose 1-superlevel set $\lbrace x \mid p(x) \geqslant 1 \rbrace$ contains a semialgebraic set $\mathcal{B} $ and has minimum volume.
Assuming that one is given a simple set $\mathcal{B} $ containing $\mathcal{K} $ and over which the integrals of polynomials can be efficiently computed, this method involves searching for a polynomial $p$ of degree $d$ which minimizes $\int_\mathcal{B}  p(x) dx$ while respecting the constraints $p(x)\geqslant 1$ on $\mathcal{K} $ and  $p(x)\geqslant 0$  on $\mathcal{B} $.
Note that the objective is linear in the coefficients of $p$ and that these last two nonnegativity conditions can be made computationally tractable by using the \textit{sum of squares relaxation}.
The complexity of the approximation depends on the degree $d$. The advantage of such a formulation lies in the fact that when the \textit{degree} of the polynomial $p$ increases, the objective value of the problem converges to the true volume of the set $\mathcal{K} $.

\begin{example}
	To better review the latter method, we illustrate it with an example for a two dimensional set given in \cite{Dabbene2017}. In order to compare the two presented methods, we also give its CAD solution.
	Having the following non-convex semi-algebraic set:
	$$\mathcal{K} = x \in \Reals^2: \begin{cases}
	f_1(x)=1+2x_2 \geqslant 0,\\
    	f_1(x)=2-4x_1-3x_2 \geqslant 0,\\
   	f_1(x)=10-28x_1-5x_2-24x_1 x_2 -18 x_2^2 \geqslant 0,\\
   	f_1(x)=1-x_2-8x_1^2-2x_1 x_2 -x_2^2-8x_1^2x_2-6x_1x_2^2\geqslant 0
	\end{cases}$$
	with a bounding box $\mathcal{B}  = [-0.8,0.6] \times [-0.6, 1.0]$,  and setting $d=8$ , the degree of the polynomial $p(x)$.
	The algorithm yields the feasible region represented in Figure \ref{fig_RegHenrion}.
	
	\begin{figure}[h]
		\begin{subfigure}[b]{0.5\textwidth}
		\includegraphics[width=\textwidth]{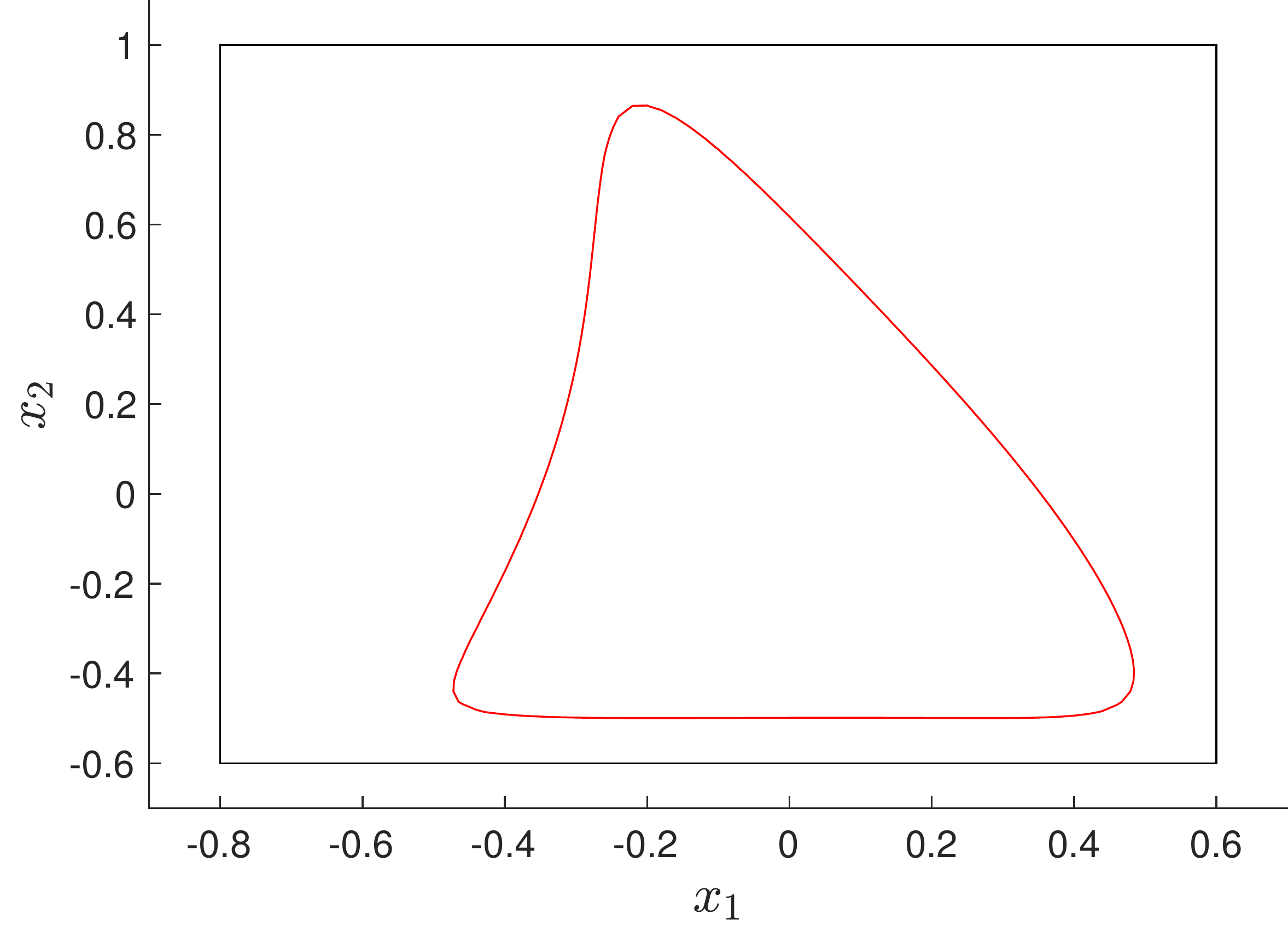}
\caption{Inner Polynomial Superlevel Set approximation of $8^\text{th}$-degree of the region $\mathcal{K}$ (the inner surface of the red line).The black rectangle represents the bounding box.}\label{fig_RegHenrion}
		\end{subfigure}
		\begin{subfigure}[b]{0.5\textwidth}
			\includegraphics[width=\textwidth]{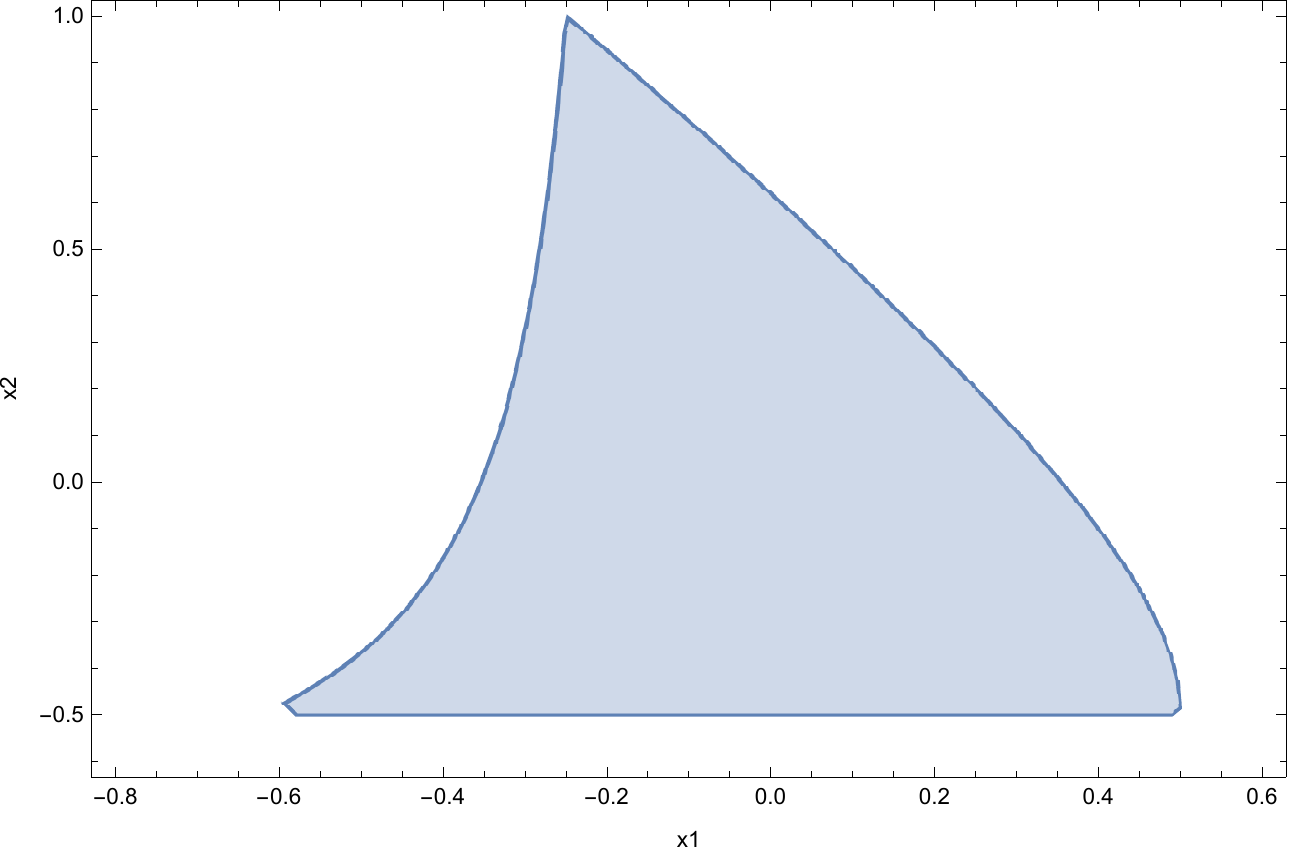}
			\caption{The region found by the CAD algorithm (the inner surface of the blue line).}\label{fig_RegHenrionCAD}
		\end{subfigure}
		\caption{The feasible regions by the two methods}
	\end{figure}
	
For the same set, even without specifying a particular box, the CAD algoririthm finds the following explicit solution:
	\begin{equation*}
	\begin{split}
&  \left(x_1=-\frac{5}{8}\land x_2=-\frac{1}{2}\right) \\
&\lor\left(-\frac{5}{8}<x_1<-\frac{1}{6}\land -\frac{1}{2}\leqslant x_2\leqslant \frac{-8
		x_1^2-2 x_1-1}{2 (6 x_1+1)}-\frac{1}{2} \sqrt{\frac{64
			x_1^4-160 x_1^3-12 x_1^2+28 x_1+5}{(6
			x_1+1)^2}}\right)\\
		&\lor \left(x_1=-\frac{1}{6}\land -\frac{1}{2}\leqslant
	x_2\leqslant \frac{7}{8}\right)\\
	&\lor \left(-\frac{1}{6}<x_1<\frac{1}{2}\land
	-\frac{1}{2}\leqslant x_2\leqslant \frac{-8 x_1^2-2 x_1-1}{2 (6
		x_1+1)}+\frac{1}{2} \sqrt{\frac{64 x_1^4-160 x_1^3-12 x_1^2+28
			x_1+5}{(6 x_1+1)^2}}\right)\\
		&\lor \left(x_1=\frac{1}{2}\land x_2=-\frac{1}{2}\right) 
	\end{split}
	\end{equation*}

As we can observe, the PSS method (Figure \ref{fig_RegHenrion}) gives us a good approximation of the feasible region, almost the same as the exact one obtained by the CAD algorithm (Figure \ref{fig_RegHenrionCAD}). However, in some cases, we observed that the PSS method may have some sensibilities when its bounding box is not well defined.
	
\end{example}

\section{Applications}\label{sec5App}
\subsection{Longitudinal dynamics of a vehicle}\label{secVehicle}
The constraints are essentials in the design of vehicle longitudinal control which aims to ensure the passenger comfort, safety and fuel/energy reduction. The longitudinal control can be designed for a highway scenario or a city scenario. In the first scenario, the vehicle velocity keeps a constant form where the main objective is the vehicle inter-distance while the second one, deals with frequent stops and accelerations, the so-called Stop-and-Go scenario \cite{Villagra2008}. The inter-distance dynamics can be represented as an single integrator driven by the difference between the leader vehicle velocity $V_l$ and the follower vehicle velocity $V_x$ , i.e., $\dot d = V_l - V_x $.\\
In this example, suppose we want to follow the leader vehicle, and stay within a fixed distance from it (measuring the distance through a camera/radar system). Additionally, suppose we enter a desired destination through a GPS system, and suppose our GPS map contains all the speed information limits. Our goal is the follower longitudinal speed $V_x$ to follow a reference speed $V_{xr}(t) \in [0, \min(V_l, V_{\max})], V_{\max} \in \Reals >0$  given by the minimum between the leader vehicle  speed and the speed limit.

The longitudinal dynamics of a follower vehicle is given by the following model:
\begin{equation}\label{ModelVehicle}
M \dot V_x(t) = \dfrac{u(t)}{r} -C_a V_x^2 (t)
\end{equation}
where $V_x$ is the longitudinal speed of the vehicle, $u$ is the motor torque, taken as \textit{control input} and the physical constants: $M$ the vehicle's mass, $r$ the mean wheel radius, and $C_a$ the aerodynamic coefficient.\\
The model is differentially flat, with $V_x$ as a flat output.
An open loop control yielding the tracking of the reference trajectory $V_{xr}$ by $V_x$, assuming the model to be perfect, is
\begin{equation}\label{ModelVehicleFlatness}
u_r (t) = r \left(  M  \dot V_{xr}(t) + C_a V_{xr}^2 (t)\right) 
\end{equation}
If we desire an open-loop trajectory $u_r \in  C^{0}$, then for the flat output, we should assign a Bézier curve of degree $d>1$.
We take $V_{xr}$ as reference trajectory, a Bézier curve of degree 4 \ie $C^{4}$-function.
\begin{align*}
V_{xr} (t) &= \sum_{i=0}^{4} a_i B_{i,4} (t), \\
V_{xr} (t_0) &= V_i, \quad V_{xr} (t_f) = V_f
\end{align*}
where the $a_i$'s are the control points and the $B_{i,4}$ the Bernstein polynomials.\\
Using the Bézier curve properties, we can find the control points of the open-loop control $u_r$ in terms of the $a_i$'s by the following steps:

\begin{enumerate}
	\item First, we find the control points  $a_i^{(1)}$ for $\dot V_{xr}$  by using the Equation \eqref{eqControlPointsDerivative}: 
	\begin{equation*}
	\dot V_{xr}= \sum_{i=0}^{3} a_i^{(1)} B_{i,3} (t)
	\end{equation*}
	
	\item We obtain the term $V_{xr}^2$ by
	\begin{equation*}
	V_{xr}^2= \sum_{i=0}^{4} a_i B_{i,4} (t)\sum_{i=0}^{4} a_i B_{i,4} (t) = \sum_{i=0}^{8} p_i B_{i,8} (t)
	\end{equation*}
	which is a Bézier curve of degree $8$ and where the control points $p_i$ are computed by the multiplication operation (see Equation \eqref{eqBezierMultip}).
	
	\item We elevate the degree of the first term up to $8$ by using the Equation \eqref{eqDegElev} and then, we find the sum of the latter with the Bézier curve for  $V_{xr}^2$. We end up with $u_r$ as a Bézier curve of degree $8$ with nine control points $U_i$:
	\begin{equation*}
	u_r (t) =  rM \dot V_{xr} + rC_a V_{xr}^2 = rM \sum_{i=0}^{3} a_i B_{i,3} (t) + rC_a (\sum_{i=0}^{4} a_i B_{i,4} )^2 = \sum_{i=0}^{8} U_i B_{i,8} (t)
	\end{equation*}
	with $U_i =r_k(a_0, \ldots,a_4)$.
\end{enumerate}

\subsubsection{Symbolic input constraints}
We want the input control points $U_i$ to be 
\begin{equation}
U_\text{min} < U_i < U_\text{max} \quad i = 0, \ldots, 8 \label{eqExampleInputs}
\end{equation}
where  $U_\text{min} = 0$ is the lower input constraint and $U_\text{max}= 10$ is the high input constraint. By limiting the control input, we indirectly constraint the fuel consumption.
The initial and final trajectory control points are defined as
$V_{x}(t_0) = a_0 = 0$ and $ V_{x}(t_1) = a_4 = 1$ respectively. 

The constraint \eqref{eqExampleInputs} directly corresponds to the semi-algebraic set:
The constraint \eqref{eqExampleInputs} corresponds to the \textit{semi-algebraic set} \ie the following system of nonlinear inequalities:
\begin{align}\label{Eq_System}
\begin{cases}
0<U_0= 4\, a_1<10 \\
0< U_1=a_1 + \frac{3\, a_2}{2} <10\\
0< U_2=\frac{4\, a_1^2}{7} - \frac{5\, a_1}{7} + \frac{12\, a_2}{7} + \frac{3\, a_3}{7} < 10\\
0< U_3=\frac{15\, a_2}{14} - \frac{10\, a_1}{7} + a_3 + \frac{6\, a_1\, a_2}{7} + \frac{1}{14} <10 \\
0< U_4=\frac{18\, a_2^2}{35} - \frac{10\, a_1}{7} + \frac{10\, a_3}{7} + \frac{16\, a_1\, a_3}{35} + \frac{2}{7}<10\\
0< U_5=\frac{10\, a_3}{7} - \frac{15\, a_2}{14} - \frac{6\, a_1}{7} + \frac{6\, a_2\, a_3}{7} + \frac{5}{7}<10\\
0< U_6=\frac{4\, a_3^2}{7} + \frac{5\, a_3}{7} - \frac{3\, a_1}{7} - \frac{9\, a_2}{7} + \frac{10}{7}  <10\\
0<U_7=\frac{5}{2} - \frac{3\, a_2}{2} <10\\
0< U_8=5 - 4\, a_3 < 10
\end{cases}
\end{align}
In order to solve symbolically the system of inequalities \ie to find the regions of the intermediate control points $a_i$, we use the Mathematica function \textit{CylidricalDecomposition}. The complete symbolic solution with three intemediate control points $(a_1, a_2,a_3)$ is too long to be included.  Since the latter is too long to be included, we illustrate the symbolic solution in the case of two intermediate control points $(a_1, a_2)$ :
\begin{equation*}
\begin{split}
&(0<a_1\leqslant 0.115563\land - a_1<a_2<1.33333) \\
&\lor \left(0.115563<a_1\leqslant 0.376808\land 0.142857 \left(-3. a_1^2+2. a_1-1\right)<a_2<1.33333\right) \\
&\lor \left(0.376808<a_1\leqslant 1.52983\land \frac{4 a_1-2}{3. a_1+4}<a_2<1.33333\right) \\
&\lor \left(1.52983<a_1<2\land 0.333333 \sqrt{15.
	a_1-17}-0.333333<a_2<1.33333\right)
\end{split}
\end{equation*}

The latter solution describing the feasible set of trajectories can be used to make a choice for the Bézier control points: "First choose $a_1$ in the interval $\left( 0, 0.115563\right] $ and then you may choose $a_2$ bigger than the chosen $-a_1$ and smaller than $1.33333$. Or otherwise choose $a_1$ in the interval $\left( 0.115563, 0.376808\right]$ and, then choose $a_2$ such that $0.142857 \left(-3a_1^2+2. a_1-1\right)<a_2<1.33333$, etc."\\

In Figure \ref{fig_VehicleRegion}, we illustrate the feasible regions for the three intermediate control points $(a_1, a_2,a_3)$ by using the Mathematica function \textit{RegionPlot3D}. We can observe how the flat outputs influences the control input \ie which part of the reference trajectory influences which part of the control input. For instance in \eqref{Eq_System}, we observe that the second control point $a_1$ influences more than $a_2$ and $a_3$ the beginning of the control input (the control points $U_0,U_1, U_2$). The previous inequalities can be used as a prior study to the sensibility of the control inputs with respect to the flat outputs.

\begin{figure}
	\begin{center}
		\includegraphics[width=3.5in]{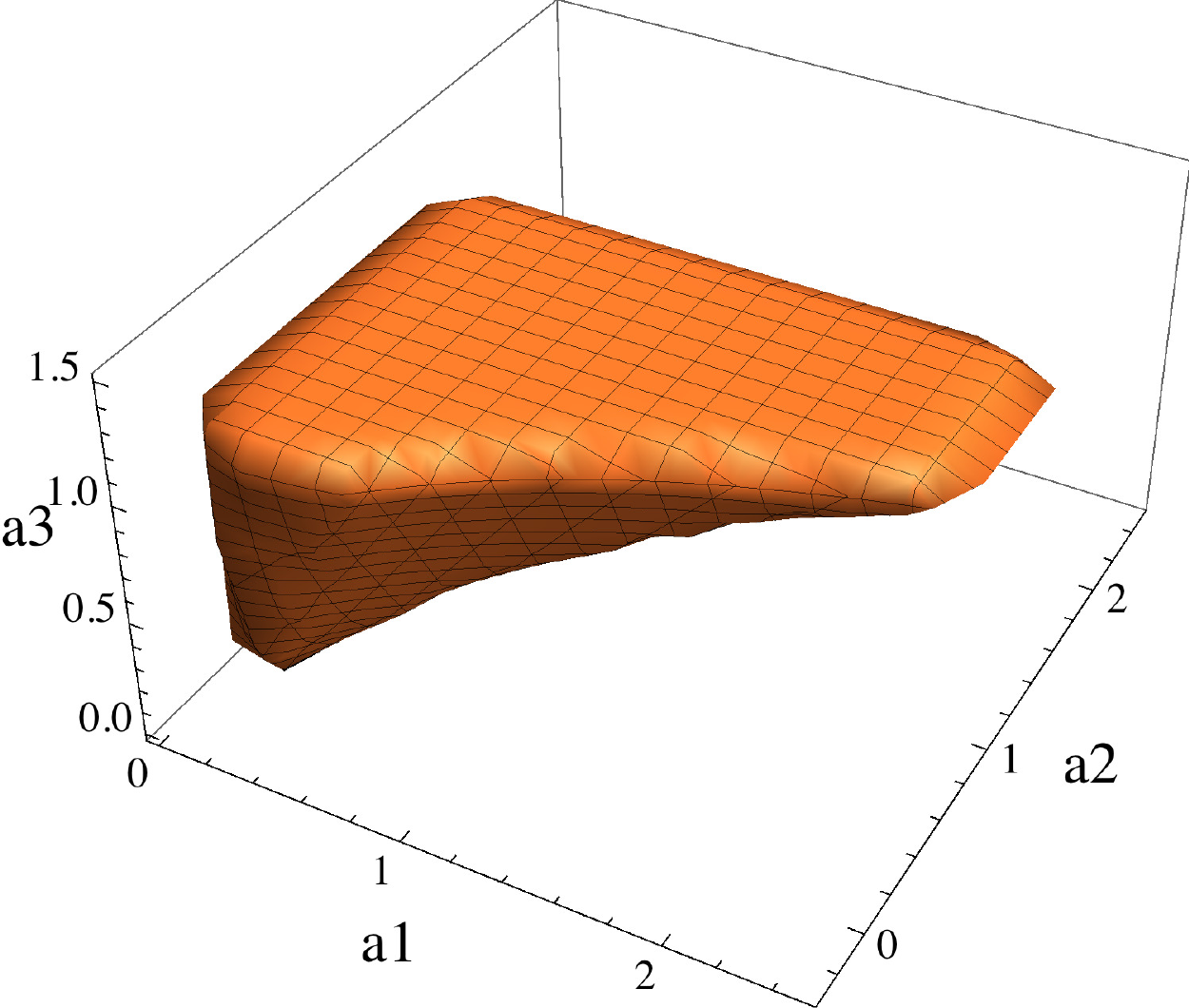}
		\caption{Feasible region for the control points of $V_{xr}$ when $U_\text{min} = 0$ and $U_\text{max}= 10$.}\label{fig_VehicleRegion}
	\end{center}
\end{figure} 

It should be stressed that the goal here is quite different than the traditional one in optimisation problems. We do not search for the best trajectory according to a certain criterion under the some constraints, but we wish to obtain the set of all trajectories fulfilling the constraints; this for an end user to be able to pick one or another trajectory in the set and to switch from one to another in the same set. The picking and switching operations aim to be really fast.

\subsubsection{Simulation results}
The proposed control approach has been successfully tested in simulation. For the physical parameters of the vehicle, academic values are chosen to test the constraint fulfilment.
For the design of the Bézier reference trajectory, we pick values for $a_1, a_2$ and $a_3$ in the constrained region. As trajectory control points for $V_{xr}$, we take the possible feasible choice $a_0= 0, a_1= 2, a_2=2.3, a_3= 1.3, a_4=1$.
Simulation results for the constrained open-loop input are shown in Figure \ref{figOpenU}.
\begin{figure}
	\begin{center}
		\includegraphics[width=3in]{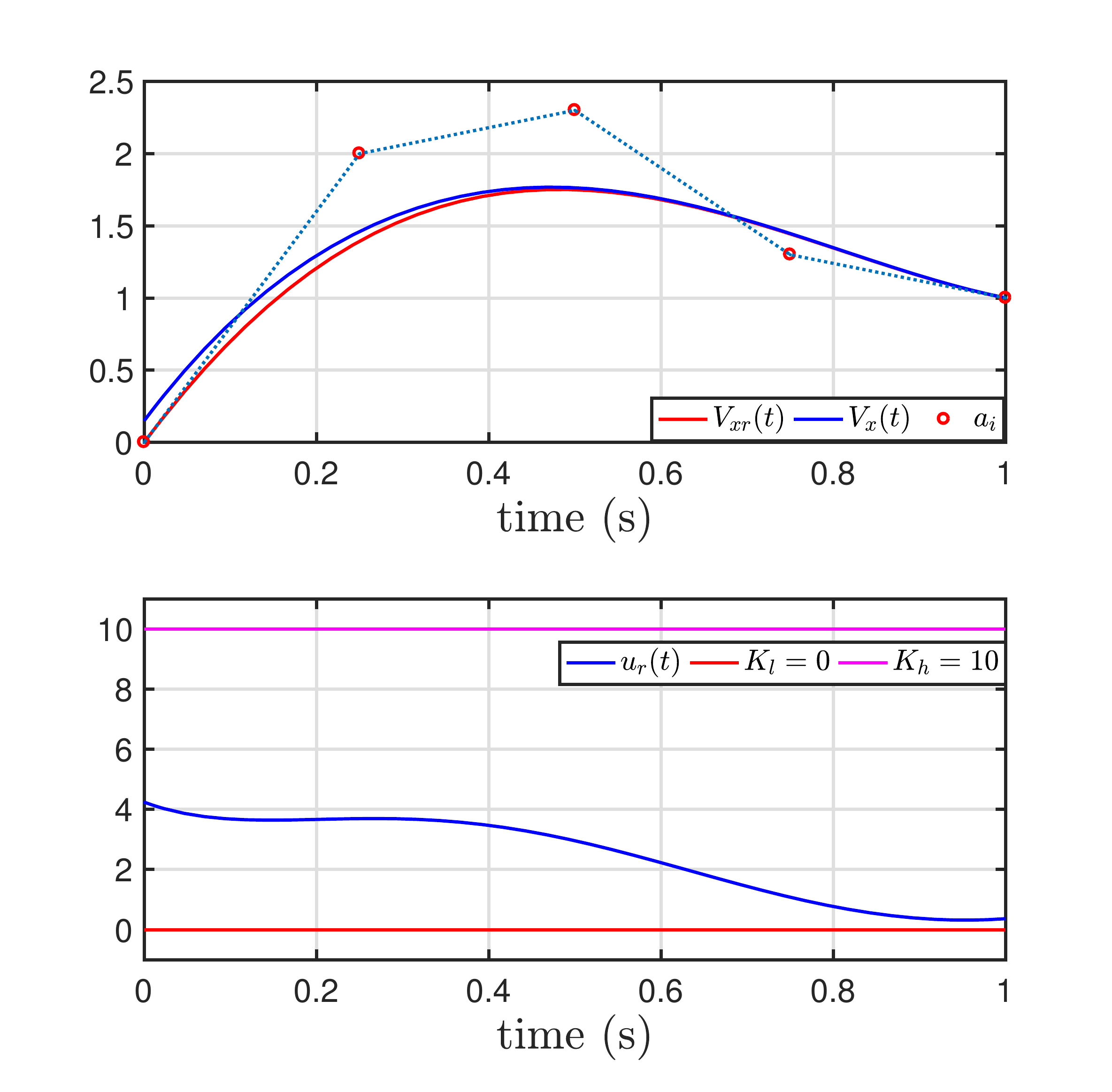}
		\caption{Open-loop input control }\label{figOpenU}
	\end{center}
\end{figure} 

The form of the closed-loop input is
\begin{equation}
u = Mr \left( \dot V_{xr} - \lambda(V_{x}-V_{xr}) \right) + r C_a V_{x}^{2}
\end{equation}
where $\lambda =9 $ is the proportional feedback gain chosen to make the error dynamics stable.
Figure \ref{FigClosedVx} shows the performance of the closed-loop control.
\begin{figure}
	\begin{center}
		\includegraphics[width=3.5in]{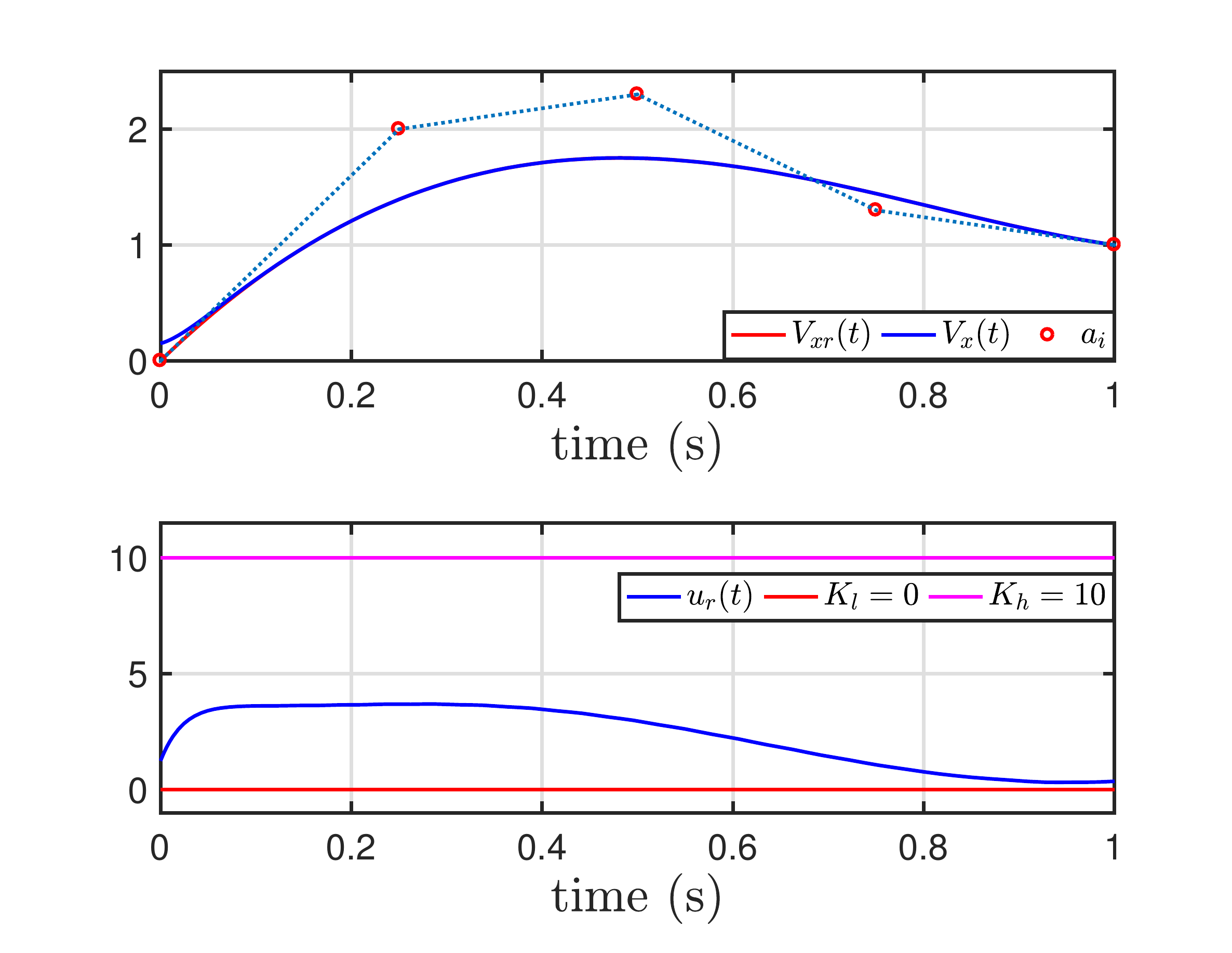}
		\caption{Closed-loop performance of trajectory tracking }\label{FigClosedVx}
	\end{center}
\end{figure} 
For both schemes, the input respects the limits.

\begin{figure}
	\begin{center}
		\includegraphics[width=3.5in]{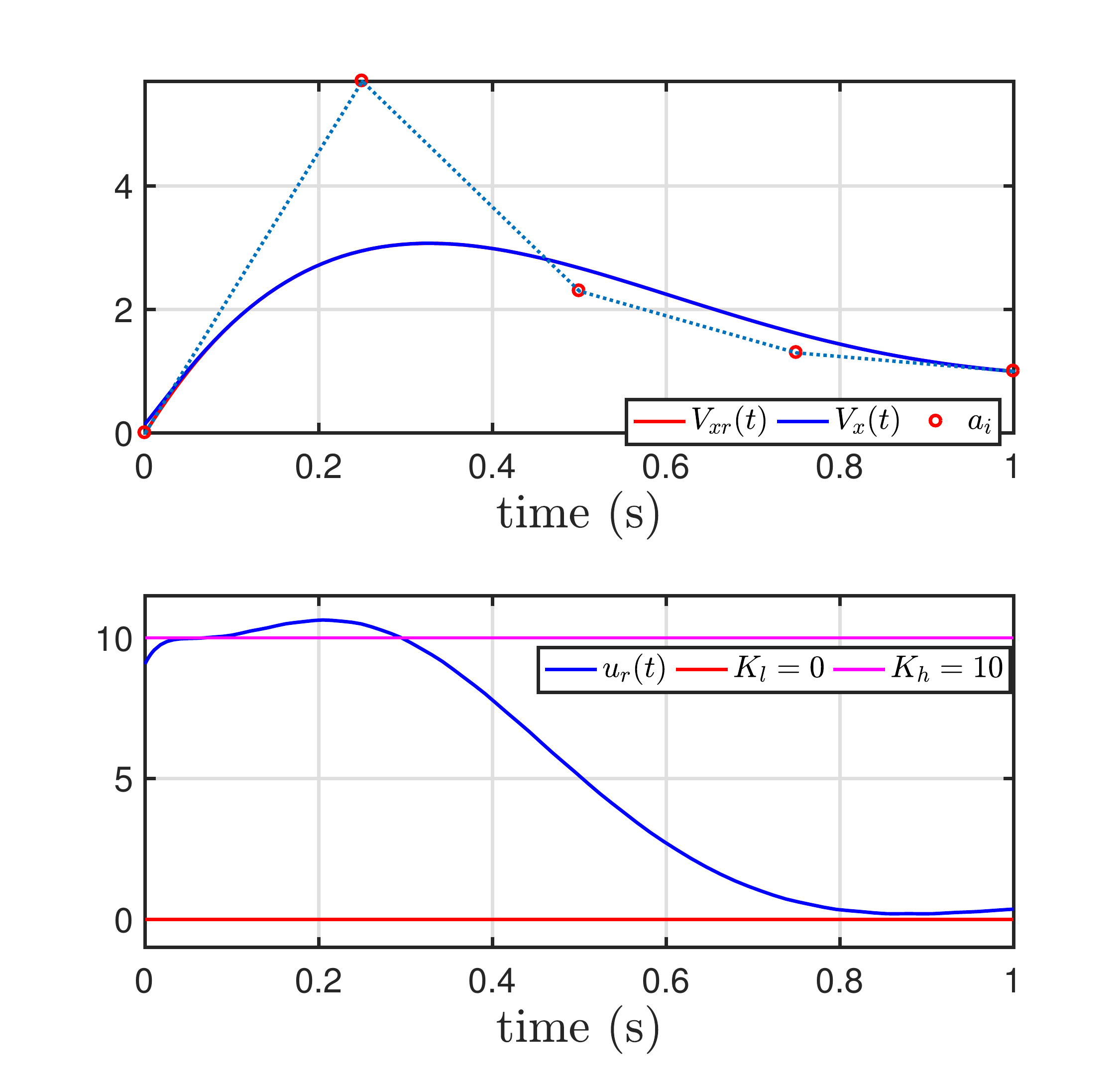}
		\caption{When control point $a_1$ is out of the its region }\label{FigClosedVxA}
	\end{center}
\end{figure} 
As shown in Figure \ref{FigClosedVxA}, choosing a control point outside of the suitable region ($a_1=5.5$) can violate the closed-loop input limits.

\subsection{Quadrotor dynamics}\label{sec2Quad}

\subsubsection{Motivation}

Over the last decade, the quadrotors have been a subject of extensive research study and have been used in a wide range of industrial and commercial applications. The quadrotors have become so popular due to their agility that allows them to hover as well as takeoff and land vertically while still being able to perform \textit{agressive trajectories} \footnote{A trajectory is considered as an aggressive one if during its tracking, one of the quadrotor motors is close to a saturation.}.

However, during aggressive trajectory design, it is difficult to ensure trajectory feasibility while trying to exploit the entire range of feasible motor inputs. Moreover, in many applications, their role is to fly in complex cluttered environments, hence there is a necessity of output constraints.
Therefore, the constraints on the inputs and states are one of the crucial issues in the control of quadrotors. 

Fortunately, with the hardware progress, today the quadrotors have speed limits of forty meters per second and more comparing to few meters per second in the past \cite{Faessler2018}. Therefore, it is important to conceive control laws for quadrotors to a level where they can exploit their full potential especially in terms of agility.

In the famous paper \cite{Mellinger2011}, is proposed an algorithm that generates optimal trajectories such that they minimize cost functionals that are derived from the square of the norm of the snap (the fourth derivative of position).
There is a limited research investigating the quadrotor constraints (see \cite{Cao2016} and the papers therein) without employing an online optimisation.

The following application on quadrotor is devoted to unify the dynamics constraints or demands constraints with the environmental constraints (\eg, fixed obstacles).

\subsubsection{Simplified model of quadrotor}

A (highly) simplified nonlinear model of quadrotor is given by the equations:
\begin{subequations}
\label{nonlinearNanoQuadModelBis}
\begin{align}
  \label{nonlinearNanoQuadModelBisOne}
  m \ddot x &= \theta u_1 \\ 
  \label{nonlinearNanoQuadModelBisTwo}
  m \ddot y &= -\phi u_1 \\
  \label{nonlinearNanoQuadModelBisThree}
  m \ddot z &= -mg + u_1 \\
  \label{nonlinearNanoQuadModelBisFour}
  I_x \ddot \theta &= u_2 \\
  \label{nonlinearNanoQuadModelBisFive}
  I_y \ddot \phi &= u_3 \\
  \label{nonlinearNanoQuadModelBisSix}
  I_z \ddot \psi &= u_4
\end{align} 
\end{subequations}
where $x$, $y$ and $z$ are the position coordinates of the quadrotor in the world frame, and $\theta$, $\phi$ and $\psi$ are the pitch, roll and yaw rotation angles respectively. The constant $m$ is the mass, $g$ is the gravitation acceleration and $I_x, I_y, I_z$ are the moments of inertia along the $y$, $x$ directions respectively. The thrust $u_1$ is the total lift generated by the four propellers applied in the $z$ direction, and $u_2, u_3$ and $u_4$ are the torques in $\theta, \phi$ and $\psi$ directions respectively. As we can notice, the quadrotor is an under-actuated system \ie it has six degrees of freedom but only four inputs. 

A more complete presentation of the quadrotor model can be found in the Section \ref{ch:6-AppQuadrotor}.

\subsubsection{Differential flatness of the quadrotor}
Here, we describe the quadrotor differential parametrization on which its offline reference trajectory planning procedure is based.
The model  \eqref{nonlinearNanoQuadModelBis} is differentially flat. Having four inputs for the quadrotor system, the flat output has four components. These are given by the vector:
\begin{equation*}
F =(x, y, z, \psi).
\end{equation*}
By equation \eqref{nonlinearNanoQuadModelBisThree}, we easily obtain expression of the thrust reference $u_{1r}$ 
\begin{align}
u_{1r} &= m ( \ddot z_r + g )
\end{align}
Then, by replacing the thrust expression in  \eqref{nonlinearNanoQuadModelBisOne}--\eqref{nonlinearNanoQuadModelBisTwo}, we obtain the angles  $\theta_r$ and $\phi_r$ given by 
\begin{subequations}
	\begin{align}\label{nonlinearThetaR}
	\theta_r &= \dfrac{m\ddot x_r}{u_{1r}} = \dfrac{\ddot x_r}{\ddot z_r + g} \\ 
	\label{nonlinearPhiR}
	\phi_r &= \dfrac{-m\ddot y_r}{u_{1r}} = \dfrac{-\ddot y_r}{\ddot z_r + g}
	\end{align} 
\end{subequations}
We then differentiate \eqref{nonlinearThetaR},
\eqref{nonlinearPhiR} and $\psi_r$ twice to obtain \eqref{nonlinearNanoQuadModelBisFour}--\eqref{nonlinearNanoQuadModelBisSix} respectively. This operation gives us $u_2$ , $u_3$ and $u_4$.
\begin{equation}\label{nonlinearNanoQuadModelParametrizationU2}
u_{2r} =    I_x \ddot \theta_r=\dfrac{I_x}{(g + \ddot z_r)}\,\left(
x_r^{(4)} - 2\, \dfrac{ x_r^{(3)} (\ddot z_r + g) - \ddot x_r z_r^{(3)})}{(\ddot z_r + g)^2}\,  
z_r^{(3)} - \dfrac{\ddot x_r z_r^{(4)}}{\ddot z_r + g}
\right),
\end{equation}

\begin{equation}\label{nonlinearNanoQuadModelParametrizationU3}
u_{3r} =    I_y \ddot \phi_r =\dfrac{I_y}{(g+ \ddot z_r)}\,\left(
-y_r^{(4)} + 2\, \dfrac{ y_r^{(3)} (\ddot z_r + g) - \ddot y_r  z_r^{(3)})}{(\ddot z_r + g)^2}\, z_r^{(3)} + 
\dfrac{ \ddot y_r z_r^{(4)}}{\ddot z_r + g}
\right),
\end{equation}
and
\begin{equation}\label{nonlinearNanoQuadModelParametrizationU4}
u_{4r} = I_z \ddot \psi_r.
\end{equation} 
A more complete model of a quadrotor and its flatness parametrization can be found in \cite{Sira-Ramirez2011} and \cite{Flores2006}.

\subsubsection{Constraints}

Given an initial position and yaw angle and a goal position and yaw angle of the quadrotor, we want to find a set of smooth reference trajectories while respecting the dynamics constraints and the environmental constraints.
Quadrotors have electric DC rotors that have limits in their rotational speeds, so input constraints are vital to avoid rotor damage. Besides the state and input constraints, to enable them to operate in constrained spaces, it is of great importance to impose output constraints.\\

We consider the following constraints:
\begin{enumerate}
	\item \textit{The thrust $u_1$} \\
	We set a maximum ascent or descending acceleration of 4g (g=9.8\,m/s$^2$), and hence the thrust constraint is defined as:
	\begin{equation}\label{eqThrustLim}
    0< u_1 \leqslant U_1^{\text{max}} = 4\,m\!\cdot\! g =20.79\,\text{N},
	\end{equation}
	where $m$ is the quadrotor mass which is set as 0.53\,kg in the simulation.
	By the latter constraint, we also avoid the singularity for a zero thrust.
	
	\item  \textit{The pitch and roll angle}  \\
	 In applications, the tilt angle is usually inferior to 14 degrees ($0.25$rad). We set
	\begin{align}\label{EqConstTheta}
	& |\phi  |\leqslant  \Phi^{\text{max}} = 0.25 \text{rad}\\
	 & |\theta | \leqslant \Theta^{\text{max}} = 0.25 \text{rad}
	\end{align}
	
	\item  \textit{The torques $u_2$, $u_3$ et $u_4$}\\
With a maximum tilt acceleration of 48\,rad/s$^2$, the limits of the control inputs are:
		\begin{align}
    & |u_2|,\,|u_3| \leqslant 48 I_{xx}= 0.3 \,\text{N$\cdot$m}\\
	&|u_4| \leqslant 48 I_{zz} =0.5 \,\text{N$\cdot$m}
	\end{align}
where $I_{xx}$, $I_{yy}$, $I_{zz}$ are the parameters of the moment of inertia, $I_{xx}\!=\!I_{yy}\!=\!6.22\!\times\!10^{-3} $kg$\,\cdot\,$m$^2$, $I_{zz}=1.12\!\times\!10^{-2}$kg$\,\cdot\,$m$^2$.
	
	\item \textit{Collision-free constraint}\\
To avoid obstacles, constraints on the output trajectory $x, y, z$ should be reconsidered.
\end{enumerate}


\paragraph{Scenario 1:} In this scenario, we want to impose constraints on the thrust, and on the roll and pitch angles.

\subsubsection{Constrained open-loop trajectory  $u_{1r}$}

We specialize the flat output $z_r$ to a sigmoid between two quasi constant altitudes, a situation frequently needed in practice:
\begin{equation}
z_r(t) = \dfrac{H_f -H_i}{2}\left( 1+ \tanh (\gamma (t - t_m))\right) + H_i
\end{equation}
where $H_i$ is the initial altitude and $H_f$ is the final altitude of the quadrotor; $\gamma$ is the slope parameter of the tanh and $t_m$ is the time when the quadrotor is taking off (see Figure \ref{FigTanh}). The maximum value for $z_r(t)$ is the final altitude $H_f$ (see fig. \ref{FigTanh}).

The easy numerical implementation of the derivatives of $z_r(t)$ is due to the nice recursion. Let $R= \tanh (\gamma (t - t_m))$ and $C=\dfrac{H_f -H_i}{2}$.  The first four derivatives of $z_r(t)$ are given as: 

\begin{align*}
\dot z_r &= \gamma C ( 1 - R^2)\\
\ddot z_r &= - 2 \gamma^2 C R( 1 - R^2)\\
z_r^{(3)} &= 2\gamma^3 C (1-R^2)(1-3 R^2)\\
z^{(4)} &= -8\gamma^4 C R(3 R^4 - 5R^2 + 2)
\end{align*}

\begin{figure}[h!]
	\begin{center}\hspace*{-15ex}
		\includegraphics[width=7in]{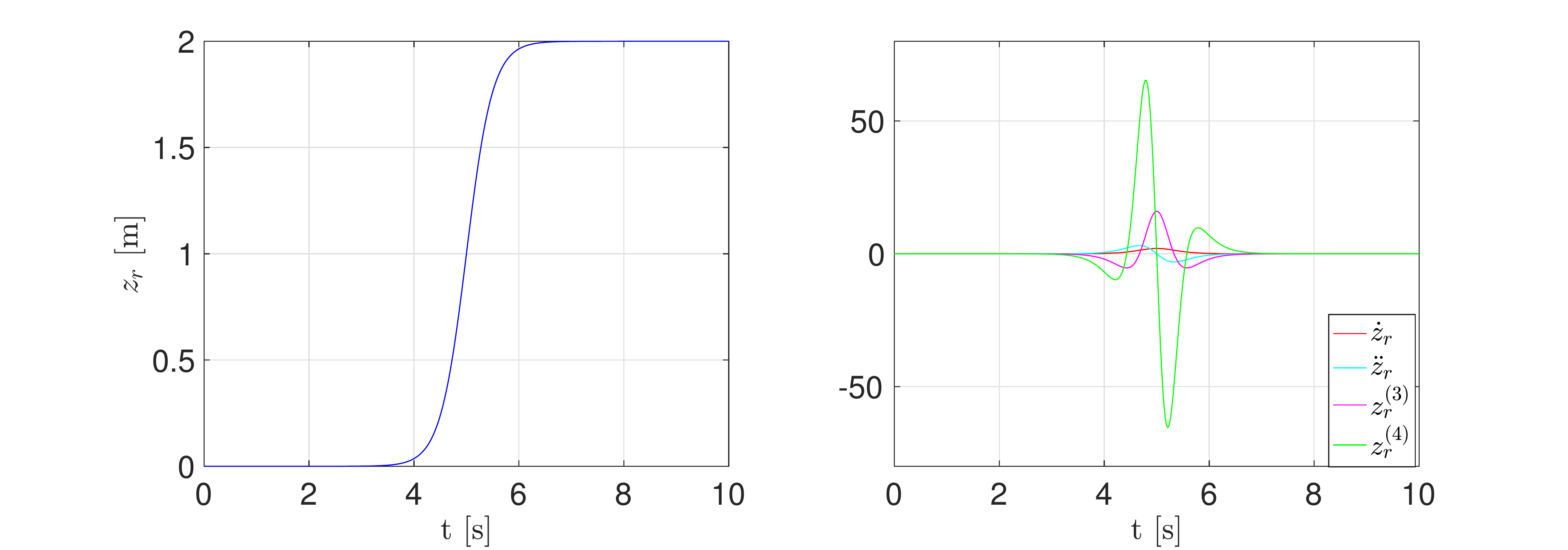}
		\caption{The reference trajectory for $z_r(t)$ (left) and its derivatives (right) with $H_i = 0$m and $H_f=2$m, $t_m = 5s$ and  parameter $\gamma =2$.}\label{FigTanh}
	\end{center}
\end{figure} 

The maximum values for its derivatives depend only on $\gamma$ and $C$, and their values can be determined.  We obtain their bounds as:
\begin{align*}
H_i \leqslant z_r &\leqslant H_f, \quad &\\
0 \leqslant \dot z_r & \leqslant b_1 \gamma C, \quad &b_1 = 1;\\
- b_2  \gamma^2 C \leqslant &\ddot z_r \leqslant b_2 \gamma^2 C, \quad  &b_2  = \dfrac{4\sqrt{3}}{9};\\
- \underline{b_3} \gamma^3  C\leqslant z_r^{(3)} &\leqslant\overline{b_3} \gamma^3 C, \quad & \underline{b_3}=\dfrac{2}{3}, \quad \overline{b_3} = 2;\\
-b_4 \gamma^4 C \leqslant z^{(4)} &\leqslant b_4 \gamma^4 C,  \quad &b_4 \approx 4.0849.
\end{align*}

Consequently, from the thrust limits \eqref{eqThrustLim},  we have the following inequality
$$
0< m(-b_2 \gamma^2 +g)  \leqslant u_{1r} = m( \ddot z_r +g) \leqslant m(b_2 \gamma^2 +g) < U_1^{\text{max}}.$$
The input constraint of $u_{1r}$ will be respected by choosing a suitable value of $\gamma$ and $C$ such that 
\begin{equation}\label{GammaInterval}
\gamma^2 C< \min\left\lbrace \dfrac{1}{b_2}\left( \dfrac{U_1^{\max}}{m}- g\right), \frac{g}{b_2}\right\rbrace.
\end{equation}

Figure \ref{FigConstrainedU1} depicts the constrained open-loop trajectory $u_{1r}$ that is well chosen by taking $\gamma =2$ and $H_f=2$m. On the other hand, in Figure \ref{FigOutConstrainedU1} is shown the violation of the thrust constraints when $\gamma =7$ is chosen out of the constrained interval \eqref{GammaInterval}.

\begin{figure}[h!]
	\begin{center}
		\includegraphics[width=3.5in]{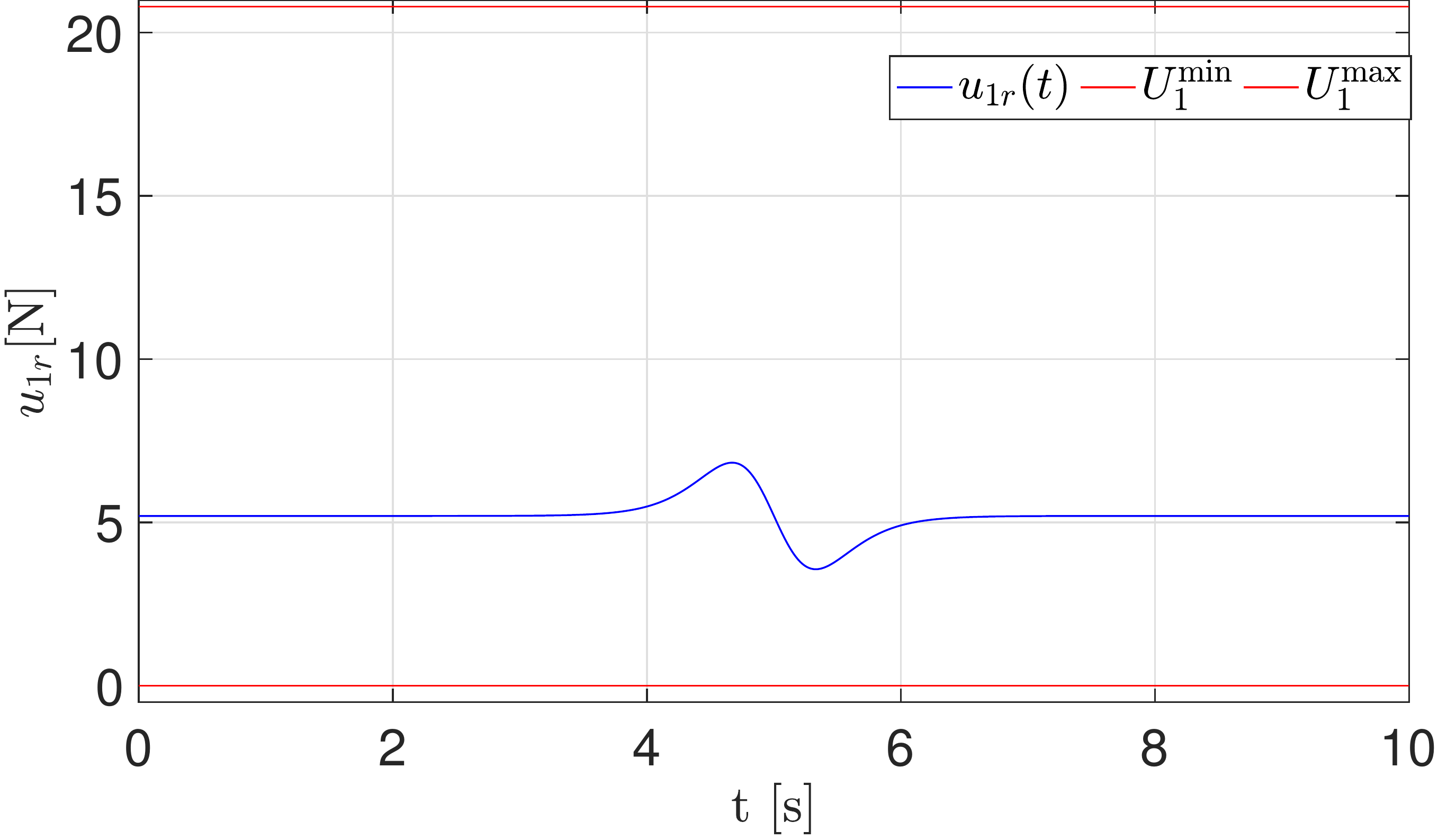}
		\caption{The reference trajectory for $u_{1r} (t)$ for a value of $\gamma =2$ and $H_f=2$m.}\label{FigConstrainedU1}
	\end{center}
\end{figure}

\begin{figure}[h!]
	\begin{center}
		\includegraphics[width=3.5in]{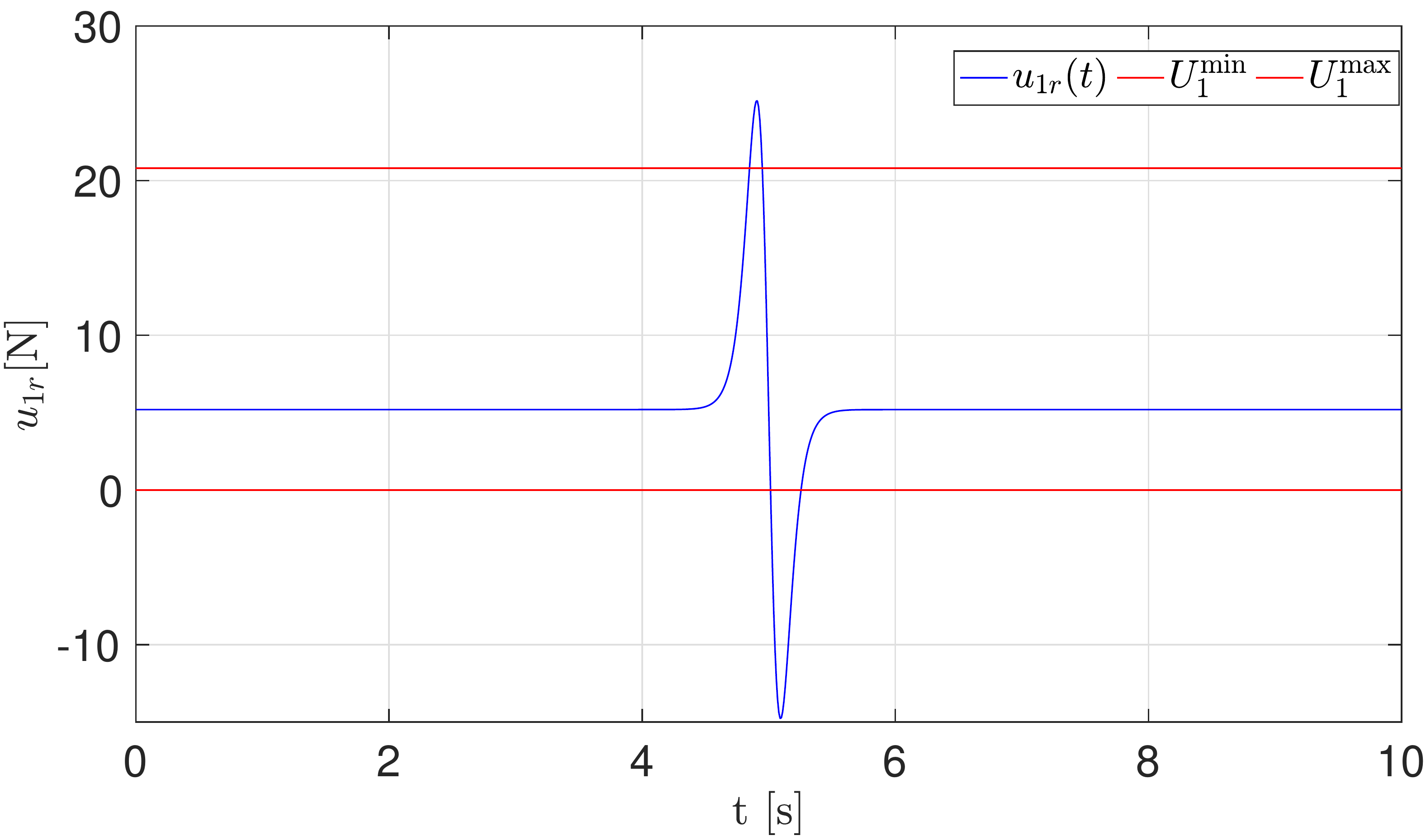}
		\caption{When the value for $\gamma$ is out of the defined interval, the constraints on the open-loop trajectory $u_{1r} (t)$ are not respected. The reference trajectory for $u_{1r} (t)$ for a value of $\gamma =7$. }\label{FigOutConstrainedU1}
	\end{center}
\end{figure}

\subsubsection{Constrained open-loop trajectories $\theta_{r}$ et $\phi_{r}$} \label{secU2R}
In the rest of the study, we omit the procedure for the angle $\phi_r$ since is the same as for the angle $\theta_r$.

\begin{enumerate}
	\item In the first attempt, the reference trajectory $x_r$ will be a Bézier curve of degree $d=6$ with a predefined control polygon form as:

$$	\Vect{A_{x}} = \left\lbrace a, a, a, \dfrac{a+b	}{2}, b, b, b \right\rbrace. $$

The aim of the first and the final control point repetitions is to fix the velocity and acceleration reference equilibrium points as : $ \dot x_r(t_0) =  \dot x_r(t_f) =0$ and $ \ddot x_r(t_0) =  \ddot x_r(t_f) =0.$

 The  control polygon of the velocity reference trajectory $\dot x $ is :
 $$	\Vect{A_{\dot x}} = \left\lbrace 0, 0, \dfrac{d}{T}\dfrac{b-a}{2}, \dfrac{d}{T}\dfrac{b-a}{2}, 0, 0 \right\rbrace.$$
 
  The  control polygon of the acceleration reference trajectory $\ddot x $ is :
 $$	\Vect{A_{\ddot x}} = \left\lbrace 0,\dfrac{d(d-1)}{T^2} \dfrac{a+b}{2}, 0, -\dfrac{d(d-1)}{T^2} \dfrac{a+b}{2}, 0 \right\rbrace. $$

The proposed form of Bézier curve provide us the explicit bounds of its second derivative $\ddot x_r $ when $a=0$ such that $ \ddot x_r^{\min}= -\frac{144}{25} \frac{b}{ T^2} $ and $\ddot x_r^{\max} =\frac{144}{25} \frac{b}{ T^2}$.

From the Equations \eqref{EqConstTheta} and \eqref{nonlinearThetaR}, we get
\begin{equation}
 \frac{-\frac{144}{25} \frac{b}{ T^2} }{b_2 \gamma^2 C+g} \leqslant \theta_{r} = \dfrac{\ddot x_r}{\ddot z_r + g} \leqslant  \frac{\frac{144}{25} \frac{b}{ T^2} }{-b_2 \gamma^2 C+g}
\end{equation}

\begin{figure}[h!]
	\begin{center}\hspace*{-15ex}
		\includegraphics[width=7in]{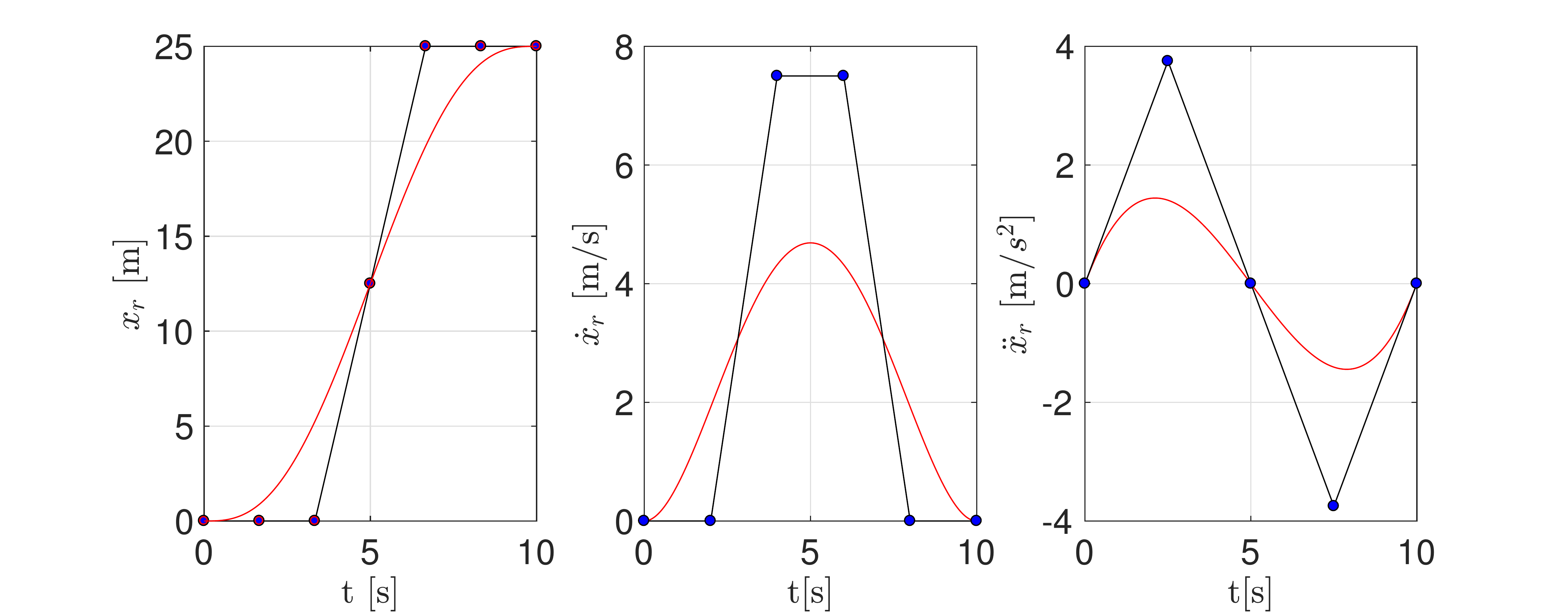}
		\caption{The Sigmoid Bézier trajectory $x_r$,  the velocity trajectory $\dot x_r$ and the acceleration trajectory $\ddot x_r$ with their respective control polygons when $a = 0$ and $b=25$. }\label{FigXcurve_caseA}
	\end{center}
\end{figure} 

\begin{figure}[h!]
	\begin{center}
		\includegraphics[width=5in]{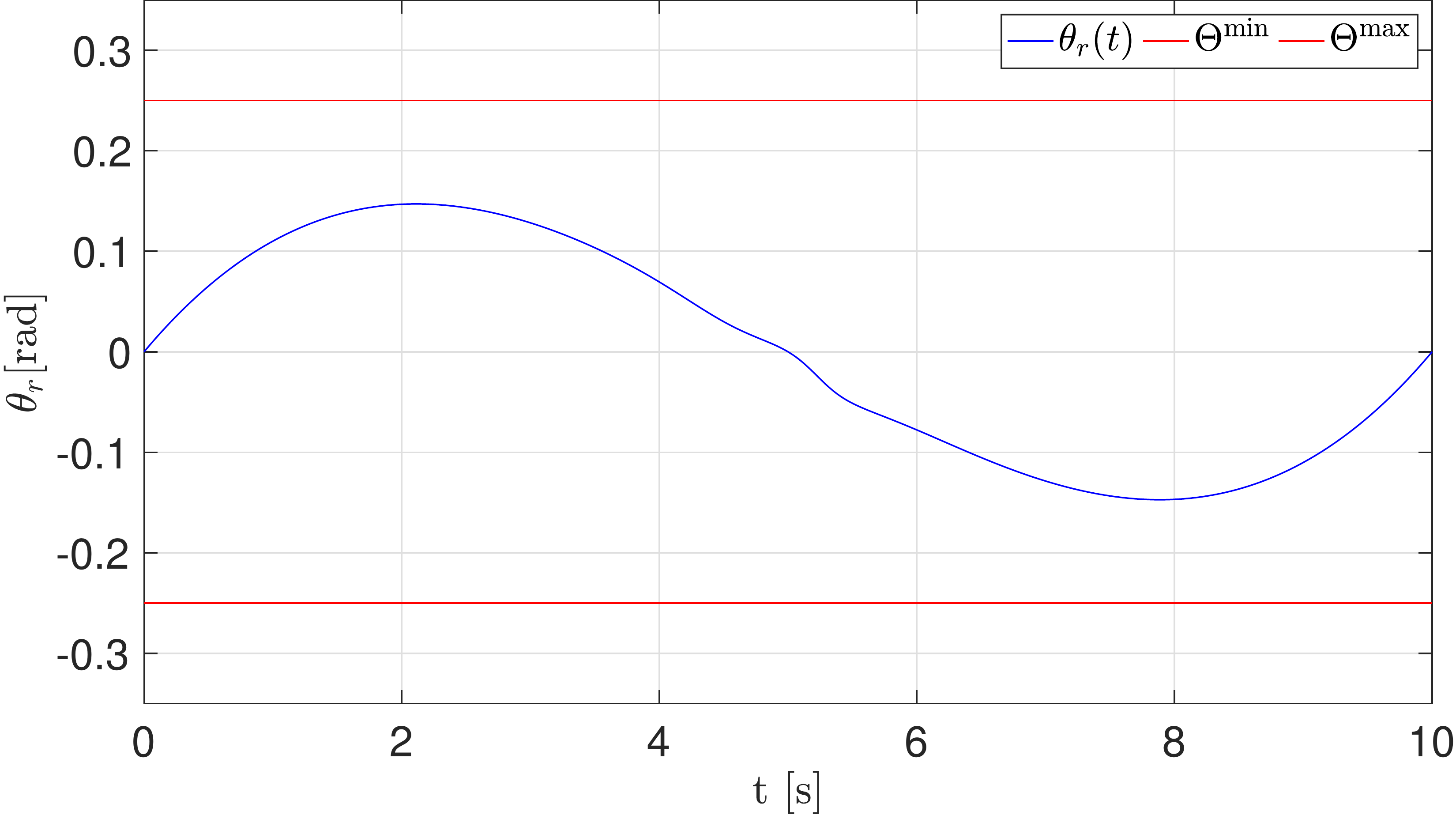}
		\caption{The open-loop trajectory $\theta_{r} (t)$  for Sigmoid Bézier trajectory}\label{FigConstrainedTheta_CaseA}
	\end{center}
\end{figure}

	\item In a second case, the reference trajectory $x_r$ can be any Bézier curve. However, we need to impose the first and last controls points in order to fix the initial and final equilibrium states. For the example, we take a Bézier trajectory of degree $d=8$ with control polygon defined as:
	
	$$	\Vect{A_{x}} = \left\lbrace a, a, a, \alpha_{1}, \alpha_{2}, \alpha_{3}, b, b, b \right\rbrace.$$
	
\end{enumerate}

When $\gamma =2$ and $H_i =0$m, $H_f =2$m are fixed, the minimum and maximum values for $\ddot z_r$ are also fixed. Therefore, to impose constraints on $\theta_r$, it remains to determine $\ddot x_r$, \ie the control points of $x_r$ 

\begin{align}\label{eqConstraintThetaX}
&\ddot x_r \leqslant  (-b_2 \gamma^2 C+g) \Theta^{\text{max}} = X^{\max} \approx 1.682 m/s^2,\\
 &\ddot x_r  \geqslant -(b_2 \gamma^2 C+g) \Theta^{\text{max}} =  X^{\min} \approx -3.222m/s^2.
\end{align}

The initial and final trajectory control points are defined as $x_{r}(t_0) = a= 0$ and $ x_{r}(t_f) = b = 2$ respectively.  
Therefore, for  $\ddot x_r$ where $T= t_f-t_0=10$, we obtain the following control polygon $\Vect{A_{ \ddot x}} = (a_{\ddot x i})_{i=0}^6$ :

 $$	\Vect{A_{\ddot x}} = \left\lbrace 0,\frac{14 \alpha_1}{25}, \frac{14 \alpha_2 - 28 \alpha_1}{25},  \frac{14 \alpha_1 - 28 \alpha_2 + 14 \alpha_3}{25}, \frac{14 \alpha_2 - 28 \alpha_3 +28}{25}, \frac{14 \alpha_3 - 28 }{25}, 0  \right\rbrace.$$

As explained in the previous section, to reduce the distance between the control polygon and the Bézier curve, we need to elevate the degree of the control polygon $\Vect{A_{ \ddot x}}$. We elevate the degree of $\Vect{A_{\ddot x}}$ up to $16$ and we obtain a new augmented control polygon $\Vect{A^{A}_{ \ddot x}}$ by using the operation \eqref{eqDegElev} (see Figure \ref{FigXcurve} (right)).

The equation \eqref{eqConstraintThetaX} translates into a system of linear inequalities \ie \textit{semi-algebraic set} defined as :
\begin{equation}
X^{\min} < a^{A}_{\ddot x i}  = f(\alpha_{1},\alpha_{2}, \alpha_{3})<X^{\max} \quad i = 0, \ldots,16.
\end{equation}

We illustrate the feasible regions for the control points by using the Mathematica function \textit{RegionPlot3D} (see Figure \ref{FigRegionQuad}).

\begin{figure}[h!]
	\begin{center}\hspace*{-15ex}
		\includegraphics[width=7in]{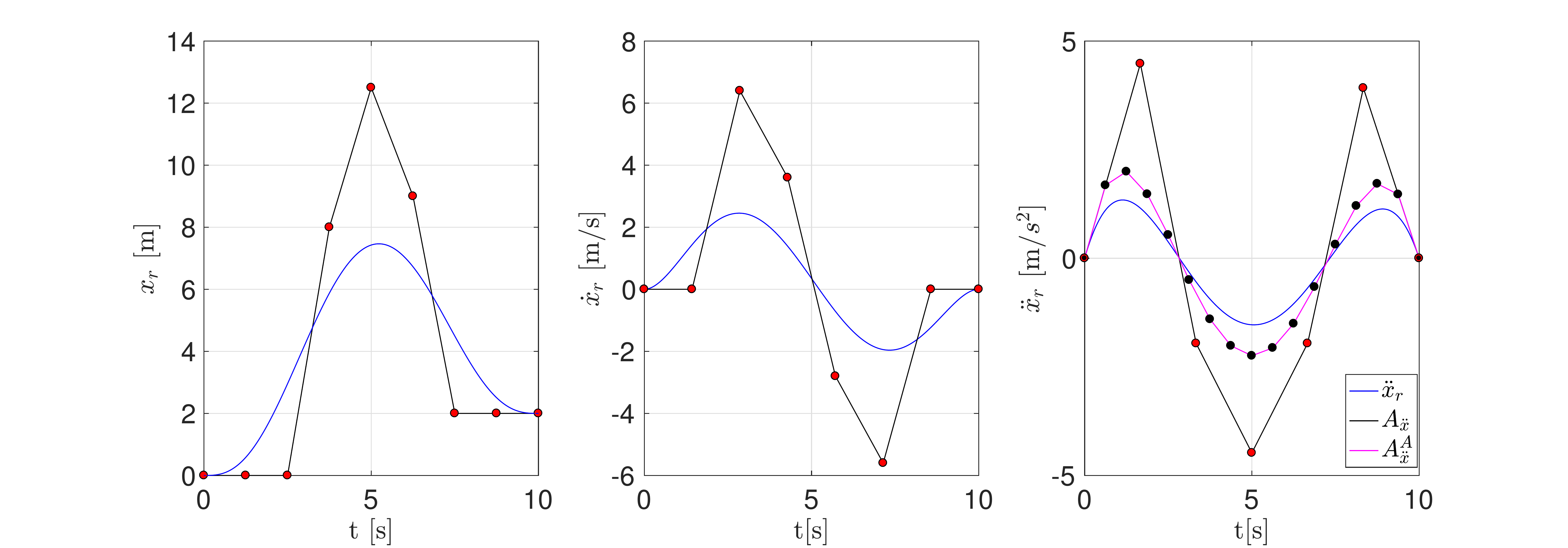}
		\caption{The Bézier curve $x_r$, $\dot x_r$, $\ddot x_r$ (blue lines) and their respective control polygons (black linear piecewise lines) with $a = 0$, $\alpha_{1} =8 $, $\alpha_{2} =12.5 $ , $\alpha_{3} = 9$ and $b=2$. The augmented control polygon for $\ddot x_2$ is represented by the magenta line.}\label{FigXcurve}
	\end{center}
\end{figure}

\begin{figure}[h!]
	\begin{center}
		\includegraphics[width=3.5in]{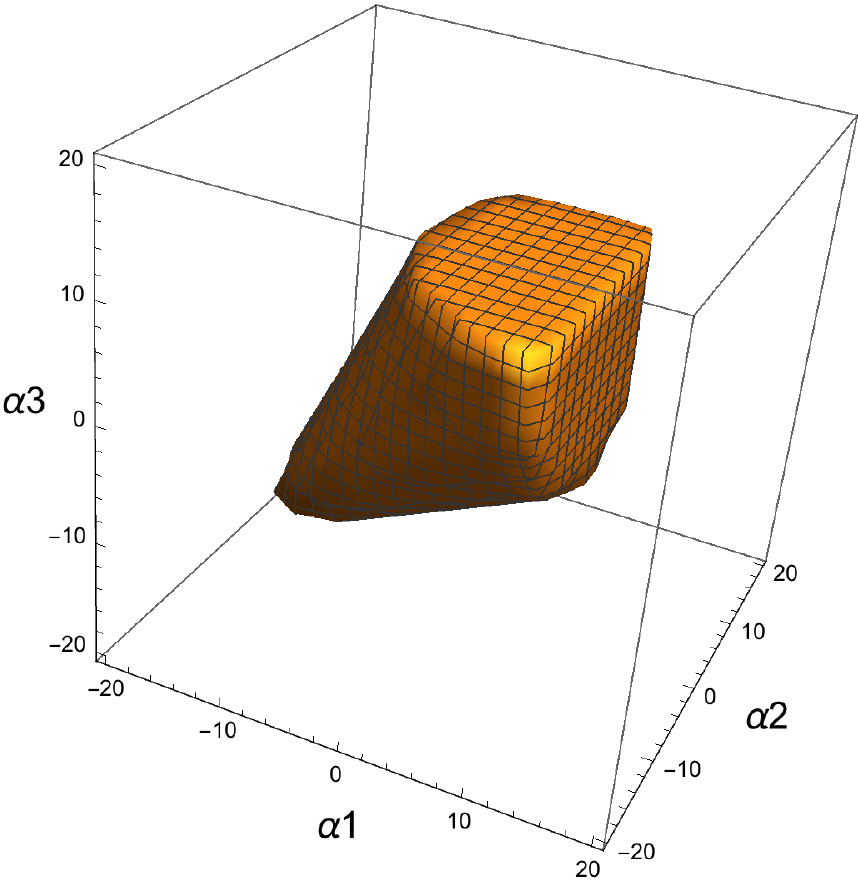}
		\caption{Feasible region for the intermediate control points of $x_r(t)$ while fulfilling the constraints on the roll angle.}\label{FigRegionQuad}
	\end{center}
\end{figure}

\begin{figure}[h!]
	\begin{center}
		\includegraphics[width=5in]{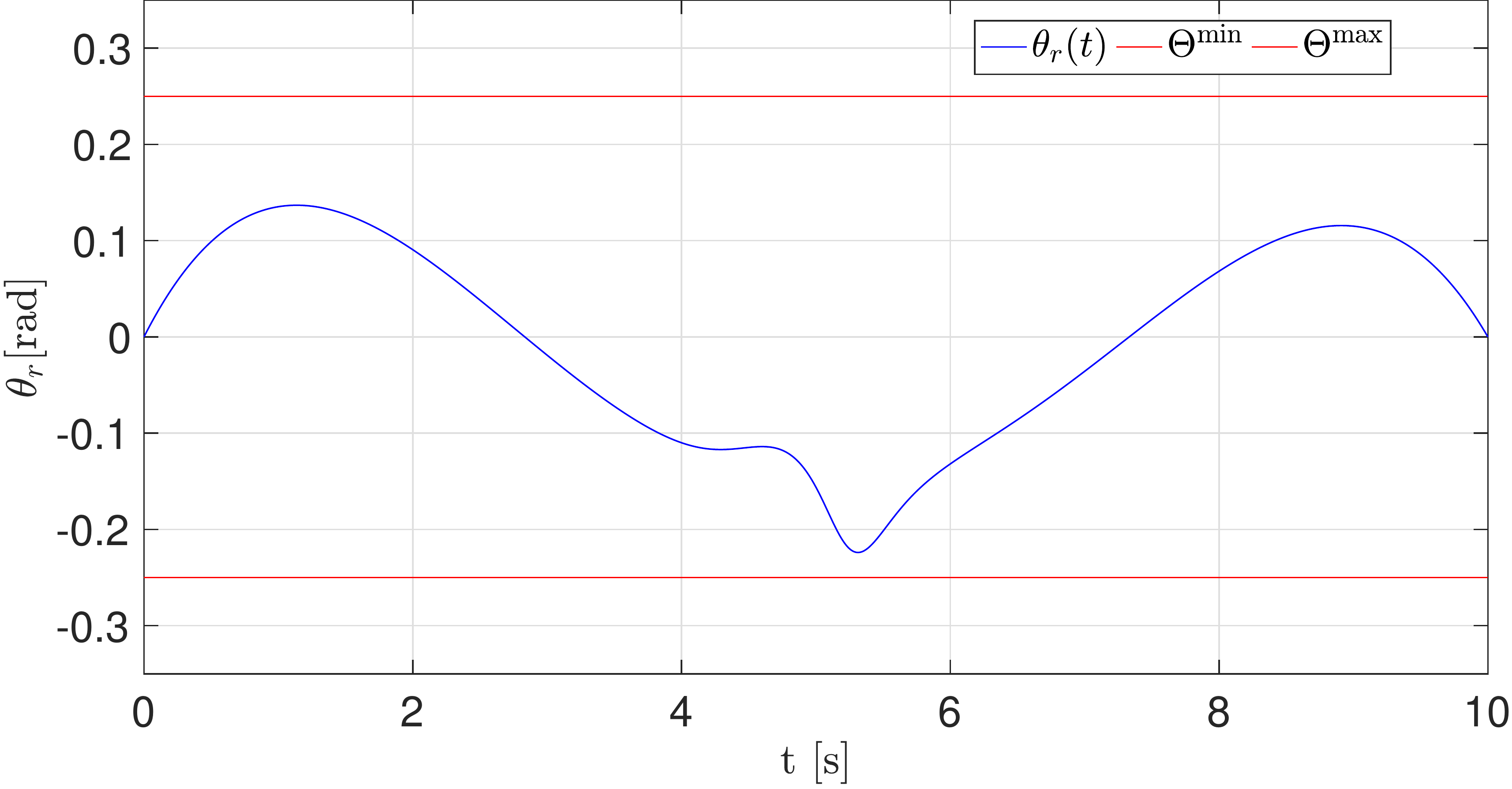}
		\caption{The constraints on the open-loop trajectory $\theta_{r} (t)$ are respected.}\label{FigConstrainedTheta}
	\end{center}
\end{figure} 

\paragraph{Scenario 2:} In this scenario, we discuss the output constraints.

\subsubsection{Constrained open-loop trajectories  $x_r$ and  $y_r$ }
Here we discuss the scenario when the quadrotor has already been take off by an initial Bézier curve that fulfils the previous input/state constraints and avoids the known static obstacles.
 Then, suddenly appear new obstacle in the quadrotor environment. To decide, whether the quadrotor should change its trajectory or continue to follow the initial trajectory, we use the quantitative envelopes of the Bézier trajectory presented in Section \ref{subSecBezierPolygon} to verify if its envelope region overlaps with the regions of the new obstacle.

We construct the quantitative envelopes for $x_r$ and $y_r$ using Section \ref{subSecBezierPolygon}. We find the maximal distance of the Bézier curve w.r.t. to the chosen control polygon. We choose as intermediate control points for $x_r$ and $y_r$ ($\alpha_1=8; \alpha_2=12.5; \alpha_3=9$ and $\beta_1=4; \beta_2=2.5; \beta_3=2$ respectively). The bounded region of the chosen reference trajectories $x_r$ and $y_r$ are depicts in Figure \ref{FigConstrainedXY_time}. 

In particular, the figure \ref{FigConstrainedXY} demonstrates the benefit of the bounded trajectory region. We can precisely determine the distance between the quadrotor pathway and the obstacles.

\begin{figure}[h!]
	\begin{center}
		\includegraphics[width=5in]{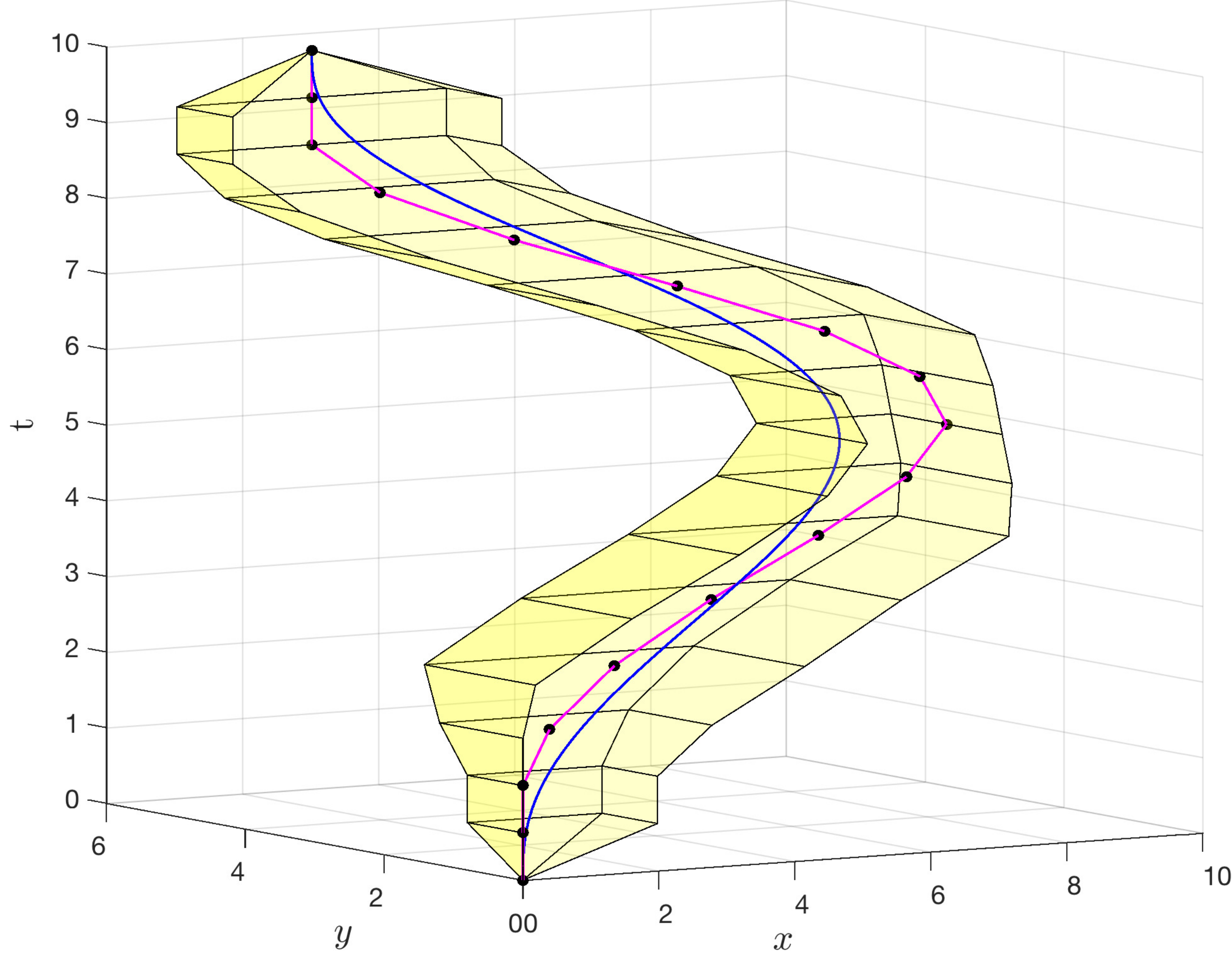}
		\caption{The constrained reference trajectories $x_{r}(t)$ and $y_{r}(t)$ and their quantitative bounded region w.r.t. to the chosen Bézier control polygon.}\label{FigConstrainedXY}
	\end{center}
\end{figure} 

\begin{figure}[h!]
	\begin{center}\hspace*{-15ex}
		\includegraphics[width=6in]{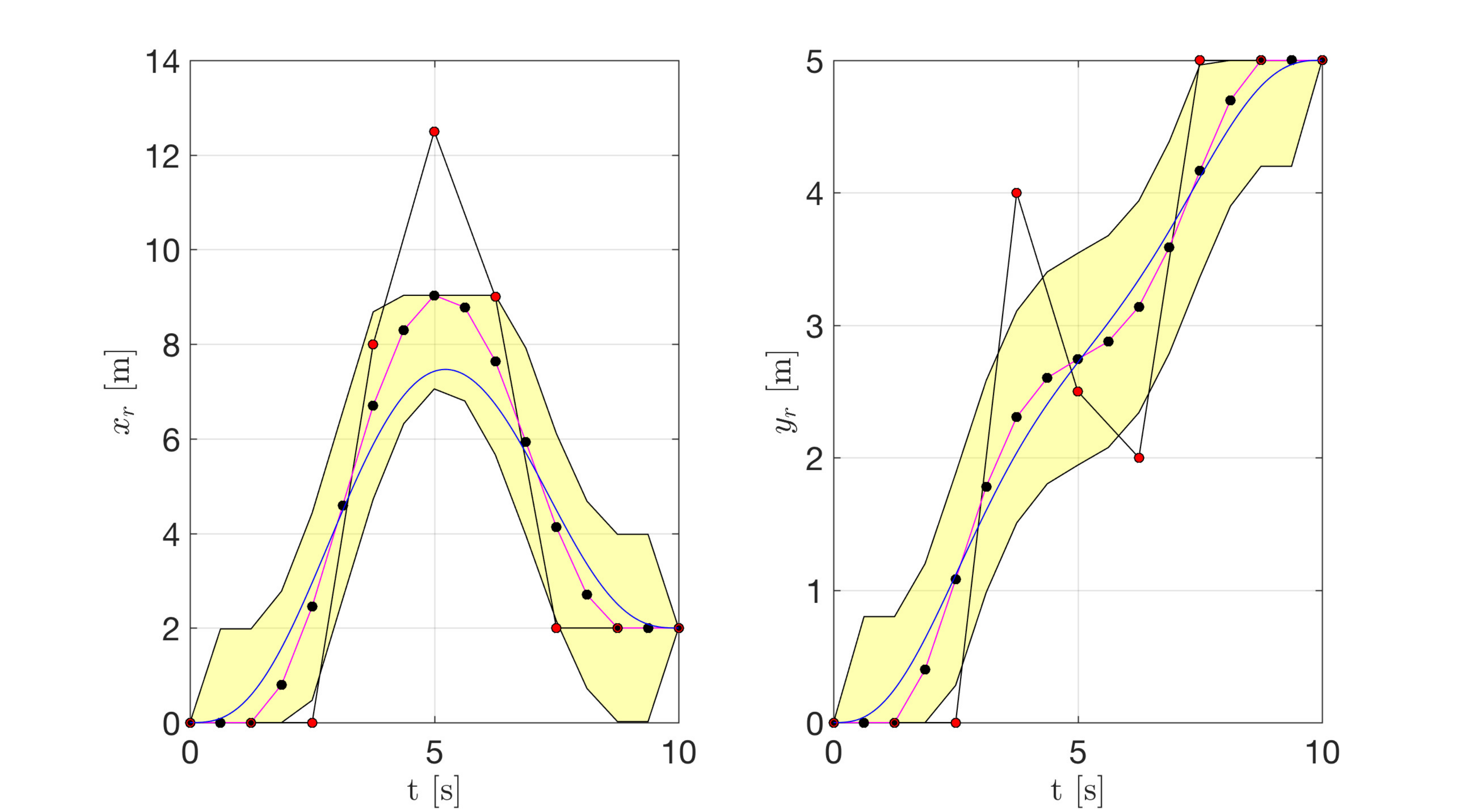}
		\caption{The quantitative envelopes for the reference trajectories $x_{r}(t)$ and $y_{r}(t)$ (the yellow highlighted regions). The augmented control polygons for $x_{r}(t)$ and $y_{r}(t)$ (magenta line). For the simulation, the intermediate control points for $x_r$ and $y_r$ are $\alpha_1=8; \alpha_2=12.5; \alpha_3=9$ and $\beta_1=4; \beta_2=2.5; \beta_3=2$ respectively. }\label{FigConstrainedXY_time}
	\end{center}
\end{figure} 

\paragraph{Scenario 3:}
In this scenario, we discuss the input constraints $u_2$ and $u_3$ when the quadrotor is in \textit{hover mode }\ie moving in a horizontal plane.
\subsubsection{Constrained open-loop trajectories  $u_2$ and $u_3$ }
By the previous constraints on $\theta_r$ and $u_{1r}$, we implicitly constrain the torque input $u_{2r}$. A more general case can also be treated if we assume that when the quadrotor reaches the desired altitude, it moves in a horizontal plane. In that case by having slow dynamics for $z_r (t)$ such that $\ddot z_r \approx 0$, we therefore have:
  \begin{subequations}
\begin{align}
u_{2r}  &=  C_x x_r^{(4)}  \\
u_{3r}  &=  C_y  y_r^{(4)} 
\end{align} 
\end{subequations}
where  $C_x =\dfrac{I_x}{g} $ and  $C_y =-\dfrac{I_y}{g}$ are constants. The latter forms a system of linear inequalities of the control points of $x_r$ and $y_r$.

\subsubsection{Constrained open-loop control for $u_{4r}$}
For $u_{4r}$, we have a simple double integrator as:
\begin{equation}
 u_{4r} = I_z \ddot \psi_r
\end{equation}
To find the regions for control points $a_{\psi i}$, we proceed in the same way as in the previous Section \ref{secU2R}.\\

\begin{remark}Our constrained trajectory reference study provides a set of feasible reference trajectories. Using the simplified models in the trajectory planning helps us to find the reference trajectory conform to the system dynamics constraints.  On the other hand, these models can not serve as a basis for the feedback law design since it will increase the uncertainties and the mismatch with the system. For that purpose, in Chapter 4,  we present the non-linear tracking of the aggressive reference trajectories by using a model-free controller. 
\end{remark}

\section{Closing remarks}
We have presented a control design for non-linear flat systems handling input/state constraints through the reference trajectory design.

The state/input constraints are translated into a \textit{system of inequalities and equalities} where the variables are the Bézier control points. This enables the input/state/output constraints to be considered into the trajectory design in a unified fashion. This allows us to develop a compact methodology to deal both with control limitations and space constraints as those arising in obstacle avoidance problems.

 The core value of this work lies in two important advantages:
\begin{itemize}
	\item The low complexity of the controller; fast real-time algorithms.
	\item The choice \ie the user can select the desired feasible trajectory. The sub-optimality may be seen as a drawback.
\end{itemize}

In the context of trajectory design, we find a successful simpler or approximated semi-algebraic set defined off-line. The closed form solution of the CAD establishes an explicit relationship between the desired constraints and the trajectory parameters. This gives us a rapid insight into how the reference trajectory influences the system behaviour and the constraints fulfillment. Therefore, this method may serve as sensitivity analysis that reflects how the change in the reference trajectory influences the input reference trajectory. Also, for fault-tolerant systems, in spirit of the papers \cite{Mai2007,Theilliol2008,Chamseddine2012, Chamseddine2012a}, this approach may be useful for the control reconfiguration when an actuator \textit{fault} occurs.

Our algorithm can deal with asymmetric constraints that may be useful in many situations e.g., for a vehicle where acceleration is created by a motor, while deceleration is achieved through the use of a mechanical brake. Increasing tracking errors and environment changes are signs that a re-planning of the reference trajectory is needed. Having the symbolic form of the exact solution, allows us a quick re-evaluation over a new range of output constraints, or with a new set of numerical values for the symbolic variables. In such case, the replanning initial conditions are equivalent to the system current state.

\appendix

\section{\label{app:1-BezierGeomSign}Geometrical signification of the Bezier operations }
Here we present the geometrical signification of the degree elevation of the Bezier trajectory $y(t)$ (Figure \ref{fig_DegElevBezier}), the addition (Figure \ref{fig_AdditionBezier}) and the multiplication (Figure \ref{fig_MultiplicationBezier}) of two Bézier trajectories.

\begin{figure}[h!]
	\centering\hspace*{-15ex}
	\includegraphics[width=7in]{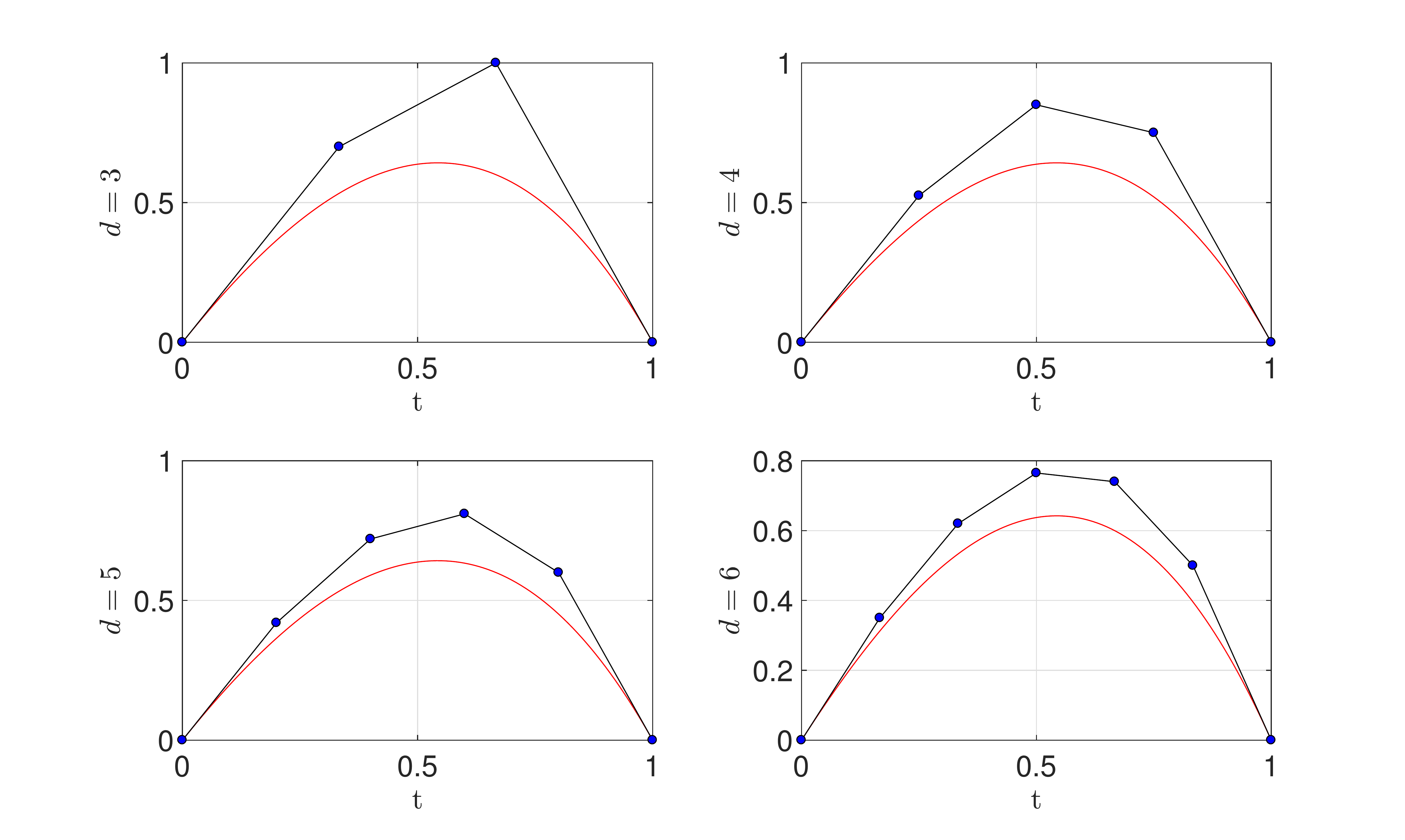}
	\caption{Degree Elevation of Bézier curve.} \label{fig_DegElevBezier}
\end{figure}

\begin{figure}[h!]
	\centering\hspace*{-15ex}
	\includegraphics[width=7in]{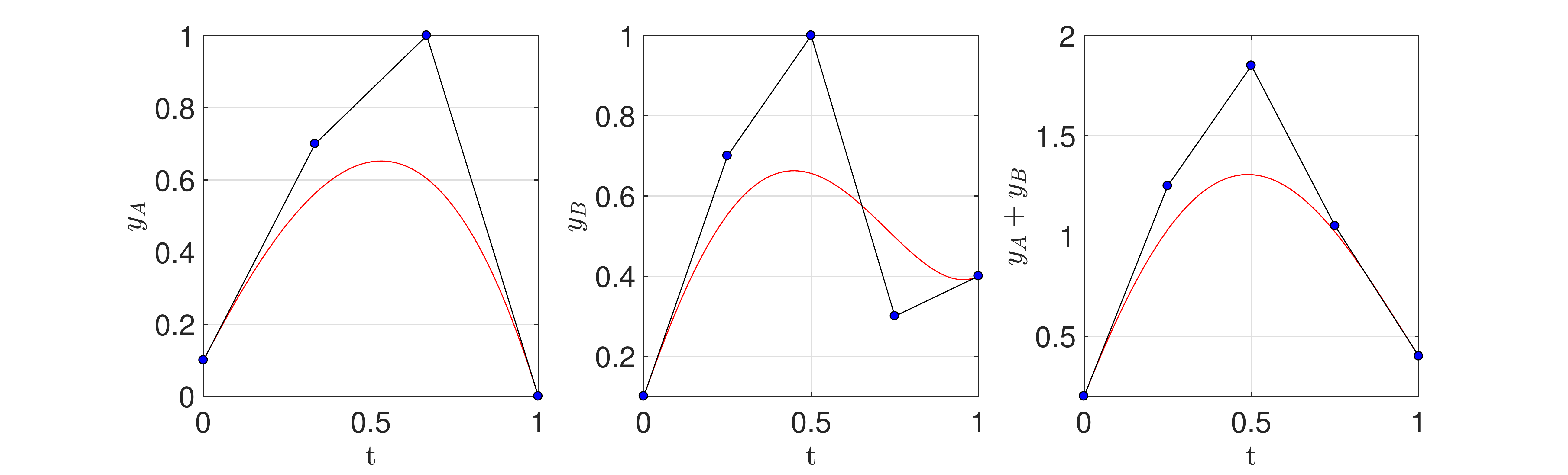}
	\caption{Addition of two Bézier curves.}\label{fig_AdditionBezier}
\end{figure}

\begin{figure}[h!]
	\centering\hspace*{-15ex}
	\includegraphics[width=7in]{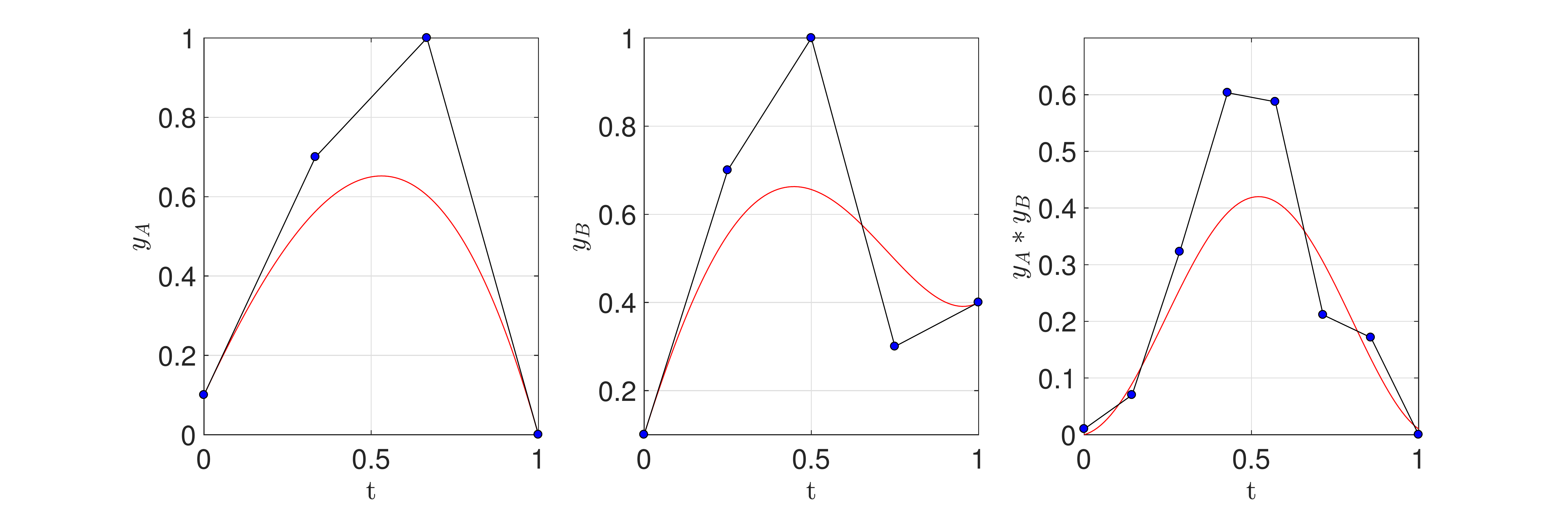}
	\caption{Multiplication of two Bézier curves.}\label{fig_MultiplicationBezier}
\end{figure}

\section{Trajectory Continuity}\label{app:1-TrajConti}

In the context of feedforwarding trajectories, the "degree of continuity" or the smoothness of the reference trajectory (or curve) is one of the most important factors. The smoothness of a trajectory is measured by the number of its continuous derivatives.
We here give some definitions on the trajectory continuity when it is represented by a parametric curve \cite{Barsky1990}.

\textbf{Parametric continuity} A parametric curve $y(t)$ is $n$-th degree continuous in parameter $t$, if its $n$-th derivative  $\frac{d ^{n}y(t)}{dt^{n}}$ is continuous. It is then also called $C^n$ continuous.

The various order of parametric continuity of a curve can be denoted as follows:
\begin{itemize}
	\item $C^0$curve \ie the curve is continuous.
	\item $C^1$curve \ie first derivative of the curve is continuous.  For instance, the velocity is continuous.
	\item $C^2$ curve \ie first and second derivatives of the curve are continuous. (The acceleration is continuous)
	\item $C^3$curve \ie first, second and third derivatives of the curve are continuous. (the jerk is continuous)
	\item $C^n$curve \ie first through $n$th derivatives of the curve are continuous.
\end{itemize}

\begin{example}
	Lets take a linear curve for the joint position of a robot, as:
	$$p(t)= p^i + \frac{p^f -p^i}{T_{\text{tt}}}t$$ where $p^{i}$ is the initial position, $p^{f}$ is the final position and $T_{\text{tt}}$ is the time interval.We obtain for the velocity and the acceleration the following curves:
	\begin{itemize}
		\item for the velocity: $v(t) =\dot p=\frac{p^f -p^i}{T_{\text{tt}}}$
		\item for the acceleration $a(t)= \ddot p = \begin{cases}
		\infty, \quad t= 0, T_{\text{tt}}\\
		0, \quad 0<t<T_{\text{tt}}
		\end{cases}$
	\end{itemize}
	In this example, we can observe infinite accelerations at endpoints and discontinuous velocity when two trajectory segments are connected.
\end{example}



\end{document}